\documentclass[11pt]{article}
\usepackage[a4paper, margin=27mm]{geometry}

\usepackage{amsmath, amssymb, amsthm}  
\usepackage{newtxtext, newtxmath}
\usepackage{microtype}
\usepackage{bm}
\usepackage{mathtools}
\usepackage{array}

\usepackage{graphicx} 

\theoremstyle{plain}
\newtheorem{theorem}{Theorem}

\newtheorem{proposition}{Proposition}
\newtheorem{lemma}{Lemma}

\theoremstyle{definition}

\begin{document}

\title{\vspace{-0.6cm}\textbf{Identifiable factor analysis for mixed continuous and binary variables based on the Gaussian-Grassmann distribution}}

\author{Takashi Arai\thanks{\texttt{takashi-arai@sci.kj.yamagata-u.ac.jp}} \vspace{0.2cm} \\ 
Faculty of Science, Yamagata University, Yamagata 990-8560, Japan}

\date{\vspace{-0.5cm}}

\maketitle

\abstract{
    We develop a factor analysis for mixed continuous and binary observed variables.
    To this end, we utilized a recently developed multivariate probability distribution for mixed-type random variables, the Gaussian-Grassmann distribution.
    In the proposed factor analysis, marginalization over latent variables can be performed analytically, yielding an analytical expression for the distribution of the observed variables. 
    This analytical tractability allows model parameters to be estimated using standard gradient-based optimization techniques.
    We also address improper solutions associated with maximum likelihood factor analysis.
    We propose a prescription to avoid improper solutions by imposing a constraint that row vectors of the factor loading matrix have the same norm for all features.
    Then, we prove that the proposed factor analysis is identifiable under the norm constraint.
    We demonstrate the validity of this norm constraint prescription and numerically verified the model's identifiability using both real and synthetic datasets.
    We also compare the proposed model with quantification method and found that the proposed model achieves better reproducibility of correlations than the quantification method.
}


\section{Introduction}
Factor analysis is a well-established method for elucidating the underlying latent structure of observed variables.
The application areas range from quantitative psychology and marketing to product management, where the subject of interest, e.g., depressive tendencies or personality traits, cannot be measured directly.
In such situations, latent variable models provide a principled framework for discovering the underlying structure of the data.
However, conventional factor analysis is designed primarily for continuous observed variables.
In today's data-rich era, categorical data, including nominal and ordinal variables, are also abundant, and conventional factor analysis cannot successfully handle such cases.
Consequently, there is a growing need to develop latent variable models that can handle mixed continuous and binary variables.

The easiest way to handle categorical variables is the method of quantification.
In quantification, categorical variables are transformed into dummy variables, which are then treated as continuous variables.
While the conditions under which such treatment is justified have been studied~\cite{Rhemtulla2012, Verhulst2021}, the quantification method is sometimes criticized as being unreflective.
Another straightforward approach is to extend the Gaussian conditional distribution given latent variables to an exponential family distribution, i.e., binary Factor Analysis~\cite{Tipping1998} or exponential family principal component analysis (PCA)~\cite{Collins2001, Mohamed2008}.
In the fields of educational measurement and psychometrics, the same model is referred to as (multidimensional) Item Response Theory (IRT)~\cite{Lord1980, Reckase2009, Bock1981}.
The approach of exponential family PCA preserves the discreteness of the categorical data, but on the other hand, it has the drawback that the distribution for observed variables cannot be expressed analytically.
This drawback complicates the numerical calculations required to estimate model parameters, especially when the number of latent dimensions is large.
This drawback makes it necessary to rely on computationally demanding and time-consuming methods for parameter estimation such as Markov chain Monte Carlo (MCMC) simulation~\cite{Patz1999, Mohamed2008}.

The purpose of this paper is twofold.
One is to extend factor analysis to include binary variables.
For this purpose, we utilize the recently developed multivariate distribution for mixed continuous and binary random variables, the Gaussian-Grassmann distribution~\cite{Arai2022}.
The Gaussian-Grassmann distribution is a natural extension of the Grassmann distribution~\cite{Arai2021}, and corresponds to a location model~\cite{Olkin1961} with parametrized mixing probabilities and the associated parametrized normal distributions.
In the Grassmann distribution and the Gaussian-Grassmann distribution, the joint, marginal, and conditional distributions are given analytically.
This analytical tractability is inherited in the proposed factor analysis, and provides an advantage in estimating model parameters.
The second purpose is to address the issues of improper solutions and non-identifiability of the factor analysis.
Improper solutions, also known as Heywood cases, often arise in maximum likelihood estimation of the factor analysis.
To be more precise, improper solutions arise when some elements of the estimated unique variances take values close to zero, resulting in the factor scores taking extreme values.
The issue of improper solutions complicates the interpretation of the latent space as well as the non-identifiability of model parameters.
We find that similar improper solutions also appear in the proposed factor analysis.
Therefore, we propose a prescription, a norm constraint prescription, for avoiding improper solutions of model parameters associated with maximum likelihood estimation of both the conventional and proposed factor analysis. 
Then, we prove that the proposed factor analysis is identifiable under the norm constraint.

This paper is organized as follows.
In Sec.~\ref{sec:fa}, we provide summary of the proposed factor analysis for mixed-type data.
The norm constraint prescription and the definition of the contribution ratio of the latent space are described in detail.
In Sec.~\ref{sec:numerical_validation}, the validity of the proposed factor analysis is demonstrated by analyzing real and synthetic datasets. 
A biplot visualization and its interpretation are given.
The sampling distribution of maximum likelihood estimates is also investigated.
Section~\ref{sec:conclusion} is devoted to conclusions.
In the appendix, a proof of the identifiability of the proposed factor analysis is given.

\section{Factor analysis for mixed-type variables} \label{sec:fa}
In this section, we summarize the proposed factor analysis for mixed continuous and binary variables.
Of course, our model reduces to the usual factor analysis when binary variables are absent.

\subsection{Summary of the results}
We denote continuous and binary observed variables by $\mathbf{x}$ and $\mathbf{y}$, respectively.
We introduce continuous latent variables denoted by $\mathbf{z}$.
Each variable is a column vector and its dimensions are $p_x$, $q$, and $p_z$, respectively.
To express the proposed factor analysis, we define the notation of index labels.
We denote the set of whole indices of binary variables as $R \equiv \{1,2,\dots, q\}$.
The set of whole indices for binary variables is divided into two parts with subscripts of $1$ or $0$, $R = (R_1, R_0)$, for variables that take a value of $1$ or $0$, respectively.
Then, we write a subvector of the binary variables taking a value of $1$ and $0$ as $\mathbf{y}_{R_1}=\bm{1}$ and $\mathbf{y}_{R_0}=\bm{0}$, respectively.
We define a $q$-dimensional bit vector $\bm{1}_{R_1}$ whose elements take a value of $0$ or $1$,
\begin{align}
    [\bm{1}_{R_1}]_s \equiv
    \begin{cases} 1, \hspace{0.3cm} \text{if} \; \; s \in R_1 \\ 
        0, \hspace{0.3cm} \text{if} \; \; s \in R_0
    \end{cases}, \; \; (s=1, 2, \dots, q),
\end{align}

The proposed factor analysis is parameterized by $(\bm{\mu}_x, \Psi, W)$ for continuous variables and $(\mathbf{b}, G)$ for binary variables.
The $p_x$-dimensional column vector $\bm{\mu}_x$ parameterizes the mean of observed continuous variables, and the $p_x \times p_x$ diagonal matrix $\Psi$ is a covariance matrix of observational noise.
By convention, we sometimes refer to $\Psi$ as the unique variance.
The $p_x \times p_z$ matrix $W$ is a factor loading matrix for continuous variables~\cite{Murphy2012}.
The $q$-dimensional column vector $\mathbf{b}$ represents a bias term for binary variables, and the $q \times p_z$ matrix $G$,
\begin{align}
    G \equiv \begin{bmatrix} \mathbf{g}_1, &  \mathbf{g}_2, & \dots, & \mathbf{g}_q \end{bmatrix}^T,
\end{align}
is a factor loading matrix for binary variables, and $\mathbf{g}_j$ is a $p_z$-dimensional column vector.

Then, we express the conditional distribution of the observed variables given the latent variable as the product of conditionally independent normal and Bernoulli distributions $\mathrm{Ber}(\cdot)$, the latter being parametrized via a logit link function, as follows:
\begin{align}
    p(\mathbf{x},\mathbf{y}|\mathbf{z})
    =&
    \mathcal{N}\bigl(\mathbf{x} \mid \bm{\mu}_x + W (\mathbf{z} - \bm{\mu}_z), \Psi\bigr)
    \prod_{j=1}^q \mathrm{Ber}\bigl(y_j \mid \mathrm{sigm}(b_{j} + \mathbf{g}_j^T (\mathbf{z}-\bm{\mu}_z)) \bigr), \notag \\
    \equiv &
    \frac{1}{(2\pi)^{p_x/2}\det\Psi^{1/2}} \exp\biggl\{-\frac{1}{2}(\mathbf{x}-\bm{\mu}_x-W (\mathbf{z} - \bm{\mu}_z))^T \Psi^{-1}(\mathbf{x}-\bm{\mu}_x-W (\mathbf{z} - \bm{\mu}_z)) \biggr\} \notag \\
    & \frac{ \exp \Bigl\{\mathbf{y}^T ( \mathbf{b} + G (\mathbf{z}-\bm{\mu}_z) ) \Bigr\}}{\prod_{j=1}^q \Bigl[ 1 + \exp \bigl\{ b_j + \mathbf{g}_j^T (\mathbf{z} -\bm{\mu}_z) \bigr\} \Bigr]},
    \label{eq:fa_conditional}
\end{align}
where $\mathrm{sigm}(\cdot)$ refers to a sigmoid function.

The distributional assumption for latent variables $p(\mathbf{z})$ is given by a mixture of Gaussian distributions with a common covariance matrix:
\begin{align}
    p(\mathbf{z}) = & 
    \sum_{R_1 \subseteq R} \pi_{R_1}(\Sigma_z) \, \mathcal{N}(\mathbf{z} \mid \bm{\mu}_z + \Sigma_z G^T \bm{1}_{R_1}, \Sigma_z), \label{eq:fa_prior} \\
    \pi_{R_1}(\Sigma_z) \equiv & 
    \frac{ \exp \Bigl( \bm{1}_{R_{1}}^T \mathbf{b} + \frac{1}{2} \bm{1}_{R_{1}}^T G \Sigma_z G^T \bm{1}_{R_{1}} \Bigr) }{\sum_{R_{1}' \subseteq R} \exp \Bigl( \bm{1}_{R_{1}'}^T \mathbf{b} + \frac{1}{2} \bm{1}_{R_{1}'}^T G \Sigma_z G^T \bm{1}_{R_{1}'} \Bigr)},
\end{align}
where the summation $\sum_{R_1 \subseteq R}$ runs over all possible index sets for binary variables.
In this paper, the distribution of the latent variables $p(\mathbf{z})$ is referred to as the prior distribution, according to convention.
It should be noted that, although it is called a prior distribution, it is part of the model, and its parameters are estimated from the observed data.
Then, the observed distribution is induced as a continuous mixture of the conditional distribution:
\begin{align}
    p(\mathbf{x}, \mathbf{y})
    =&
    \pi_{R_1}(\Sigma_z) \, \mathcal{N}(\mathbf{x} \mid \bm{\mu}_x + W \Sigma_z G^T \mathbf{y}, \Sigma_x), \label{eq:fa_induced} \\
    \Sigma_x =& \Psi + W \Sigma_z W^T. \label{eq:fa_continuous_covariance_matrix}
\end{align}
This is a location model in which the class probabilities $\pi_{R_1}(\Sigma_z)$ are expressed as an Ising model, the location parameter are shifted by dummy variables, and covariance parameters are shared across all composite categories.
When observed variables consist exclusively of binary variables, the observed distribution is exactly the same form as the Ising model~\cite{Ising1925}, where $\mathbf{b}$ is a bias term and $\frac{1}{2} G \Sigma_z G^T$ is a weight term of the Ising model.

A posterior distribution for $\mathbf{z}$ is simply given by a normal distribution:
\begin{align}
    p(\mathbf{z}|\mathbf{x}, \mathbf{y}) =& \mathcal{N}(\mathbf{z} \mid \mathbf{m}, \Sigma_{z|x}) , \label{eq:fa_posterior} \\
    \mathbf{m} =& \bm{\mu}_z + \Sigma_{z|x} \bigl\{ W^T \Psi^{-1}(\mathbf{x}-\bm{\mu}_x) + G^T \mathbf{y}\bigr\} , \label{eq:factor_score} \\
    \Sigma_{z|x} =& \bigl(\Sigma_z^{-1} + W^T \Psi^{-1} W \bigr)^{-1}. \label{eq:factor_score_variance}
\end{align}
We call $\mathbf{m}$ in the above expression a factor score~\cite{Murphy2012}.
We can give a natural interpretation between prior and posterior distributions for the latent variables $\mathbf{z}$.
When there is no information on the binary variables $\mathbf{y}$, we represent the uncertainty of the latent variable, the prior distribution for $\mathbf{z}$, by arranging $2^q$ normal distributions with equal covariance.
By observing the binary variables, the posterior distribution reduces to single normal distribution out of a mixture of $2^q$ normal distributions.
We give the expression for the joint distribution for future reference:
\begin{align}
    p(\mathbf{x}, \mathbf{z}, \mathbf{y})
    =& 
    \pi_{R_1}(\Sigma_z) \, \mathcal{N}(\mathbf{x} \mid \bm{\mu}_x + W(\mathbf{z}-\bm{\mu}_z), \Psi)
    \, \mathcal{N}(\mathbf{z} \mid \bm{\mu}_z + \Sigma_z G^T \mathbf{y}, \Sigma_{z}). \label{eq:fa_joint}
\end{align}
More detailed expressions for the proposed factor analysis with missing values are given in Appendix~\ref{sec:appendix_fa}.

The parameter $\Sigma_z$ can be renormalized to the redefinition of the parameters $G'=G \Sigma_z^{1/2}$ and $W'=W \Sigma_z^{1/2}$, and the scale transformation of the latent variable $\mathbf{z}'=\Sigma_z^{-1/2} \mathbf{z}$ and $\bm{\mu}_z' = \Sigma_z^{-1/2} \bm{\mu}_z $.
By contrast, the parameter $\bm{\mu}_z$ is irrelevant to the representability of the model and only affects the interpretation of the latent variable, since the likelihood function does not depend on $\bm{\mu}_z$.
Therefore, in the remainder of this paper, we set $\Sigma_z = I$ without loss of generality, and also set
\begin{align}
	\bm{\mu}_z = - \frac{1}{2} \sum_{j=1}^{q} \mathbf{g}_j.
\end{align}

As an existing model of factor analysis for binary variables, there exists binary Factor Analysis~\cite{Tipping1998}.
In the fields of educational measurement and psychometrics, the same model is referred to as multidimensional IRT~\cite{Lord1980, Reckase2009, Bock1981}.
Binary Factor Analysis can be viewed as an example of a more general framework of exponential family PCA~\cite{Collins2001, Mohamed2008}, although mathematically it is more appropriate to call it factor analysis rather than principal component analysis.
Exponential family PCA uses an exponential family distribution for the conditional distribution of observed variables given the latent variable $\mathbf{z}$.
For example, in binary Factor Analysis, the conditional distribution for observed variables is given by a Bernoulli distribution with the logit link function, $p(\mathbf{y}|\mathbf{z},\theta)= \prod_{j=1}^q \mathrm{Ber}(y_j \mid \mathrm{sigm}( w_0 + \mathbf{w}_j^T \mathbf{z}))$.
This choice of the conditional distribution is the same as in our model.
By contrast, the prior distribution for the latent variable is given by a normal distribution with zero mean and unit covariance, $p(\mathbf{z})=\mathcal{N}(\mathbf{z} \mid \bm{0},I)$, unlike our model.
This introduction of the Gaussian prior distribution, however, has the drawback that the marginalization of the latent variable cannot be performed analytically, i.e., the induced distribution cannot be expressed analytically.
This drawback causes difficulty in estimating model parameters.
Hence, one has to resort to an approximation technique for parameter estimation such as variational expectation-maximization algorithm~\cite{Tipping1998}, which approximates the functional form of the posterior distribution $p(\mathbf{z}|\mathbf{y})$, or a computationally-demanding and time-consuming Monte Carlo simulation such as MCMC methods~\cite{Patz1999, Mohamed2008}.

The main difference between our model and previous studies lies in the introduction of a mixture of Gaussian distributions as a prior distribution $p(\mathbf{z})$ for the latent variables, while the proposed factor analysis itself is originally based on the Gaussian-Grassmann distribution.
In this mixture prior, the mixing weights are parametrized by the Ising model, and the means of the normal distributions are shifted by the parameter $G$.
The introduction of the mixture prior enables closed-form marginalization over the latent variables and allows the induced observed distribution to be expressed analytically.
Hence, our model has the advantage that model parameters can be estimated by maximum likelihood through standard gradient-based optimization techniques.
However, calculating the model's partition function requires summing over all possible states for binary variables, which increases exponentially with the number of binary variables.

\subsection{Improper solutions in factor analysis}
It is quite important to mention the instability of model parameters in maximum likelihood estimation for usual factor analysis.
In the usual factor analysis, the covariance matrix of the induced distribution is represented by a low-rank perturbation of a diagonal matrix, where the covariance matrix is given by the unique variance $\Psi$ plus the contribution from the latent space $W W^T$, as shown in Eq.~(\ref{eq:fa_continuous_covariance_matrix}).
However, maximum likelihood estimates are often numerically unstable, e.g., the estimated parameters may change drastically as the number of latent dimensions varies.
Specifically, some elements of the unique variance can take values as close to zero as possible.
Then, from the Eq.~(\ref{eq:factor_score}), the absolute value of some elements of the factor scores take quite large values, leading to numerical instability.
As a result, the interpretation of the factor scores becomes difficult.
In some research fields, such a problem has been recognized as improper solutions~\cite{Gerbing1985, Gerbing1987}, or Heywood cases, of maximum likelihood estimation in factor analysis.
Although the causes of such instability have been investigated and various prescriptions for avoiding the instability have been proposed~\cite{Cooperman2022}, they are somewhat technical and do not seem to offer a fundamental solution.
We believe this instability is the reason why factor analysis is not as widely used as probabilistic/nonprobabilistic PCA in practice~\cite{Price2006, Yamaguchi2008}, even though probabilistic PCA is a special case of factor analysis and makes the unnatural assumption that the variance of observational noise for all observed variables is the same, i.e., homoscedastic.
That is, the results of probabilistic PCA depend on a scale transformation of observed variables.
We believe that this dependence on the scale transformation is not a desirable property for a data analysis method, even though in practice each observed variable is often standardized to approximately meet the homoscedasticity assumption.

We therefore propose a way to avoid the instability in maximum likelihood factor analysis.
We consider that the cause of the improper solutions is the excessive degrees of freedom in the model parameters.
Improper solutions arise when a specific observed variable can be fully reconstructed by the latent variables, causing its unique variance to become zero.
Therefore, by imposing the constraint that the norm of the row vector of $\Psi^{-1/2} W$ is the same for all features, we prevent specific features from being fully reconstructed by the latent space.
To be more precise, we express the factor loading matrix $W$ using a dimensionless factor loading matrix $\tilde{W} \equiv \Psi^{-1/2} W$ as follows:
\begin{align}
    W = \Psi^{1/2} \tilde{W} = c \Psi^{1/2} \hat{W},
\end{align}
where $c$ is the norm of the row vector of $\tilde{W}$, and we have defined the normalized factor loading matrix as $\hat{W}$ whose row vectors are all normalized to one.
In this constraint, the covariance matrix $\Sigma_x$ is expressed as
\begin{align}
    \Sigma_x =& \Psi + W W^T
    = \Psi + (c\Psi)^{1/2} \hat{W} \hat{W}^T (c\Psi)^{1/2}, \notag \\
    =& (1 + c^2) \Psi \; \biggl(1-\frac{c^2}{1 + c^2} \biggr) + (1+c^2) \Psi^{1/2} \hat{W} \hat{W}^T \Psi^{1/2} \; \biggl( \frac{c^2}{1 + c^2} \biggr).
\end{align}
From the above equation, we see the relation $(1 + c^2) \Psi = \mathrm{diag}(\Sigma_x)$.
Noting that $\mathrm{diag}(WW^T) = c^2 \Psi$, the covariance matrix of the observed variables can also be expressed as follows:
\begin{align}
	\Sigma_x = & \mathrm{diag} \bigl(W W^T \bigr) \frac{1}{c^2} + W W^T.
\end{align}
The coefficient $c \ge 0$ controls the strength of the influence of the latent space.
The fraction $c^2 /(1+c^2)$ represents the proportion of the variance of the observed variables that can be explained by the latent variable.
This norm constraint on the factor loading matrix $W$ allows for a clear distinction in the role of the parameters: $W$ is used exclusively to reconstruct the covariance among observed variables, while the unique variance $\Psi$ is used exclusively to account for the variance of observational noise.

We found that the proposed factor analysis for binary variables also suffers from the instability of model parameters similar to that of the continuous variables.
In this case, some elements of $\mathbf{b}$ take a large negative value, i.e. the sign of some elements of $\mathbf{b}$ is negative and their absolute value become large, as well as some elements of the corresponding rows of $G$ take a large value.
Hence, as in the case of continuous variables, we impose the constraint on the binary factor loading matrix $G$ that each row vector has the same norm $c$:
\begin{align}
    G = & c \, \hat{G},
\end{align}
where we have defined the normalized factor loading matrix $\hat{G}$ whose row vectors are all normalized to one.

In the case of factor analysis for mixed continuous and binary variables, a similar norm constraint prescription is applied to avoid improper solutions.
We impose norm constraints on both the factor loading matrices for continuous and binary variables as
\begin{align}
    & \Psi^{-1/2} W = c \hat{W}, \hspace{1cm}
    G = c \, \hat{G},
\end{align}
where the norm $c$ takes a common value for continuous and binary variables.
The above equation is equivalent to imposing norm constraints on each row vector of the following combined factor loading matrix $M$,
\begin{align}
    M = c \hat{M} = c
    \begin{bmatrix}
        \hat{W} \\
        \hat{G}
    \end{bmatrix}, \hspace{0.5cm} 
    \hat{W} \equiv \begin{bmatrix} \hat{\mathbf{w}}_1^T; \; \hat{\mathbf{w}}_2^T; \; \dots; \; \hat{\mathbf{w}}_{p_x}^T \end{bmatrix}
    , \hspace{0.5cm} 
    \hat{G} \equiv \begin{bmatrix} \hat{\mathbf{g}}_1^T; \;  \hat{\mathbf{g}}_2^T; \; \dots; \; \hat{\mathbf{g}}_q^T \end{bmatrix}, \label{eq:factor_loading_vector}
\end{align}
where the semicolon denotes the vertical concatenation of the row vectors, $[\hat{\mathbf{w}}_1^T; \; \hat{\mathbf{w}}_2^T] \equiv [\hat{\mathbf{w}}_1, \; \hat{\mathbf{w}}_2]^T$.
Here, we shall call each row vector of normalized factor loading matrices, $\hat{\mathbf{w}}_j^T$ and $\hat{\mathbf{g}}_j^T$, a normalized factor loading vector.

\subsection{Identifiability}
When interpreting factor analysis models, it is crucial that the models are identifiable.
As in the case of the usual factor analysis, the factor loading matrices $W$ and $G$ have rotational and sign reversal symmetry on the latent space.
In fact, the likelihood function is invariant under the rotational transformation of the combined factor loading matrix $M' = M R$, where $R$ is a rotation matrix.
Although Ref.~\cite{Rubin1956} discusses the identifiability of factor analysis for continuous variables, it only mentions the conditions for the estimated parameters, and does not provide general conditions.
While there is a straightforward method of counting the degrees of freedom for model parameters~\cite{Murphy2012}, this represents a necessary condition, not a sufficient one, and thus may not work well.

Recently, it has been proven that continuous factor analysis with the norm constraint on the factor loading matrix is identifiable when an appropriate condition on the number of latent dimensions is imposed~\cite{Arai2025}.
Similarly, in factor analysis for mixed continuous and binary variables with the norm constraint, the following theorem demonstrates that the model is identifiable.
\begin{theorem} \label{theorem:1}
    Let the observed variables $\mathbf{x}$ and $\mathbf{y}$ be generated by the proposed factor analysis of Eq.~(\ref{eq:fa_induced}) with the following model parameters,
    \begin{align}
        \Sigma_x =& \Psi + W W, \\
        \bm{\mu} =& \bm{\mu}_x + W G^T \mathbf{y}, \\
        \Sigma_y =& 2 B + G G^T.
    \end{align}
    If the dimension of the latent space satisfies $q \le p_x + p_y$, the row norms of the dimensionless factor loading matrix $M \equiv [\Psi^{-1/2} W; G]$ are equal for all features, $\mathrm{diag}(M M^T) = c^2 I$, the symmetric matrix $M^T M$ is diagonalized with its nonzero and nondegenerate eigenvalues sorted in descending order $M^T M = \mathrm{diag}(\omega^2) = \Omega^2$, $\omega_i > \omega_j$, $(i < j)$, and the row sums of the matrix $M$ are nonnegative,
    \begin{align}
        \sum_{i=1}^{p_x + p_y} [M]_{ij} \ge 0, \hspace{0.5cm} j \in \{1, 2, \dots, q \}, \label{eq:G_sign}
    \end{align}
    then the model parameters are identifiable.
\end{theorem}
The proof of the theorem~\ref{theorem:1} is provided in the appendix, along with the proof of identifiability for the case involving only binary variables.

When the matrix $M^T M$ is not diagonalized, the rotation matrix $R$ to satisfy the above orthogonality condition can be constructed by arranging the eigenvectors in columns, $R=[\mathbf{v}_1, \mathbf{v}_2, \dots, \mathbf{v}_{p_z}]$, where $\mathbf{v}_s$ is a column eigenvector of the matrix $M^T M$.

\subsection{Contribution ratio}
In this subsection, we define the contribution ratio of the latent space.
We define the contribution ratio of the latent space based on its contribution to the reconstruction error of the model parameters, measured by the squared Frobenius norm.
Since the factor loading matrix $W$ depends on the scale of the variables, it is convenient to define scale-independent (dimensionless) parameter by dividing each row of the matrix $W$ by the square root of the corresponding unique variance $\Psi^{1/2}$.
We define a combined factor loading matrix as 
\begin{align}
    M = \begin{bmatrix} \Psi^{-1/2} W; \, G \end{bmatrix},
\end{align}
where the semicolon denotes the vertical concatenation of the matrix.
By using the combined factor loading matrix $M$, the parameters of the proposed factor analysis are collectively represented by $M M^T$ as shown in Appendix~\ref{sec:appendix_proof_2}.
When we approximate the combined factor loading matrix by the truncated singular value decomposition (SVD):
\begin{align}
	M =& U \Omega V^T = \sum_{i=1}^D \omega_i \mathbf{u}_i \mathbf{v}_i^T + \sum_{i=D + 1}^{p_z} \omega_i \mathbf{u}_i \mathbf{v}_i^T, \notag \\
	=& M' + \sum_{i=D + 1}^{p_z} \omega_i \mathbf{u}_i \mathbf{v}_i^T
	=
	U \Omega' V^T + \sum_{i=D + 1}^{p_z}  \omega_i \mathbf{u}_i \mathbf{v}_i^T,
\end{align}
where $D$ is an integer satisfying $1 \le D \le p_z$, $\omega_i$ is a singular value, and $\mathbf{u}_i$ and $\mathbf{v}_i$ are left and right singular vectors, respectively, it can be shown that the squared Frobenius norm is given by
\begin{align}
	|| M - M' ||_F^2 =& \mathrm{Tr} \bigl[ (M - M') (M - M')^T \bigr], \notag \\
	=&
	\mathrm{Tr} \bigl[ U \Omega^2 U^T - U \Omega'^2 U^T \bigr], \notag \\
	=&
	\mathrm{Tr} \Biggl[ \sum_{i=D+1}^{p_z} \omega_i^2 \mathbf{u}_i \mathbf{u}_i^T \Biggr]
	=
	\sum_{i=D + 1}^{p_z} \omega_i^2,
\end{align}
where $\mathrm{Tr}$ denotes the matrix trace and we have used the relation $\mathrm{Tr}[ \mathbf{u}_i \mathbf{u}_i^T ] = 1$.
From the above equation, we see that the sum of the eigenvalues $\omega_s^2$ for the discarded eigenspace represents the squared Frobenius norm of the matrix difference $M - M'$.
Therefore, the eigenvalues of the matrix $M^T M$, $\omega_s^2$, can be used to define the contribution and cumulative contribution ratio of the latent space $P_s$ and $C_s$, respectively, as
\begin{align}
    P_s \equiv \frac{\omega_s^2}{\sum_{t=1}^{p_z} \omega_t^2}, \hspace{1cm}
    C_s \equiv \sum_{t=1}^s P_t, \label{eq:contribution_ratio}
\end{align}
where the axis of the latent space, the latent factor dimension, is sorted in descending order of the contribution ratio $P_s$.

When all observed variables are continuous, the contribution ratio can also be interpreted as the fraction of the variance of observed variables explained by the latent space.
This interpretation can be seen from the conditional distribution given latent variables.
From Eq.~(\ref{eq:fa_conditional}), we define
\begin{align}
    p(\mathbf{x} \mid \mathbf{z}) =& \mathcal{N}(\mathbf{x} \mid \bm{\eta}_x(\mathbf{z}) = \bm{\mu}_x + W \mathbf{z}, \Psi), \\
    \tilde{\bm{\eta}}_x(\mathbf{z}) \equiv & \Psi^{-1/2} \bm{\eta}_x(\mathbf{z}),
\end{align}
where the observational noise of the variable $\tilde{\bm{\eta}}_x(\mathbf{z})$ is scaled to have unit variance.
From this linear relation, we see that when the latent variable follows the standard normal distribution $\mathbf{z} \sim \mathcal{N}(\mathbf{z} \mid \bm{0}, I)$, the variance of the linear combination of the latent variable $\tilde{\bm{\eta}}_x(\mathbf{z})$ is calculated as follows:
\begin{align}
    \mathrm{Cov}[ \tilde{\bm{\eta}}_x ]
    =& \mathrm{E}\bigl[ \Psi^{-1/2} (\bm{\eta}_x - \bm{\mu}_x) (\bm{\eta}_x - \bm{\mu}_x)^T \Psi^{-1/2} \bigr], \notag \\
    = & \Psi^{-1/2} W \mathrm{E}[\mathbf{z} \mathbf{z}^T] W^T \Psi^{-1/2}, \notag \\
    \equiv & \tilde{W} \tilde{W}^T.
\end{align}
The total variance of $\tilde{\bm{\eta}}_x(\mathbf{z})$ is the sum of the diagonal elements of the above covariance matrix, which is equivalent to the sum of the eigenvalues of $\tilde{W} \tilde{W}^T$:
\begin{align}
	\mathrm{Tr} \bigl[ \mathrm{Cov}[ \tilde{\bm{\eta}}_x ] \bigr] =&
	\mathrm{Tr} \bigl[ \tilde{W} \tilde{W}^T \bigr] 
	= 
	\mathrm{Tr} \Biggl[ \sum_{i=1}^{p_z} \omega_i^2 \mathbf{u}_i \mathbf{u}_i^T \Biggr]
	=
	\sum_{i=1}^{p_z} \omega_i^2,
\end{align}
where $\omega_s^2$ and $\mathbf{u}_s$ are eigenvalues and eigenvectors of the matrix $\tilde{W} \tilde{W}^T$, respectively.

Here, we transform the covariance matrix $\Sigma_x$ into a more familiar form, in terms of communalities.
We define the diagonal matrix $C$ by the row norm of the dimensionless factor loading matrix $\tilde{W} = \Psi^{-1/2} W$ as
\begin{align}
    C = [\mathrm{diag}(\tilde{W} \tilde{W}^T) ]^{1/2}.
\end{align}
Using the matrix $C$, we express the covariancee matrix $\Sigma_x$ as
\begin{align}
    \Sigma_x = \Psi + \Psi^{1/2} C \hat{W} \hat{W}^T C \Psi^{1/2}.
\end{align}
From the above equation, we read $\mathrm{diag}(\Sigma_x) = (I + C^2) \Psi$.
Then, the covariance matrix is represented by
\begin{align}
    \Sigma_x =& 
    \frac{I}{I + C^2} \mathrm{diag}(\Sigma_x) + \mathrm{diag}(\Sigma_x)^{1/2} \sqrt{\frac{C^2}{I + C^2}} \hat{W} \hat{W}^T \sqrt{\frac{C^2}{I + C^2}} \mathrm{diag}(\Sigma_x)^{1/2}, \notag \\
    =&
    \bigl( I - H^2 \bigr) \mathrm{diag}(\Sigma_x) + \mathrm{diag}(\Sigma_x)^{1/2} H \hat{W} \hat{W}^T H \mathrm{diag}(\Sigma_x)^{1/2}, \\
    H \equiv & \mathrm{diag}(h_j) \equiv \sqrt{\frac{C^2}{I + C^2}}, \hspace{0.5cm} j \in \{1, 2, \dots, p_x \}.
\end{align}
By convention, we refer to $h_j$ as the communality.
The communality satisfies $0 \le h_j \le 1$ and represents the fraction of the correlations in observed variables that is accounted for by the latent factors.
Thus, our norm constraint corresponds to the constraint that the communalities for scaled variables $\mathrm{diag}(\Sigma_x)^{-1/2} \mathbf{x}$ are the same across all features.

We briefly mention the correspondence between conventional factor analysis and non-probabilistic PCA.
For this purpose, we consider the limit where the unique variances approach zero, $\Psi= (I - H^2) \mathrm{diag}(\Sigma_x) \rightarrow O$.
This limit is equivalent to taking the limit of $H \rightarrow I$ and $C \rightarrow \infty I$.
In this case, the observed variables have one-to-one correspondence with the latent factors.
In fact, in the limit where the observational noise is zero, the covariance matrix of the posterior distribution, Eq.~(\ref{eq:factor_score_variance}), approaches to zero:
\begin{align}
    \Sigma_{z|x} =& (\Sigma_z^{-1} + W^T \Psi^{-1} W)^{-1} = (\Sigma_z^{-1} + C \hat{W}^T \hat{W} C)^{-1} \rightarrow O.
\end{align}
Furthermore, factor scores are obtained as the projection of the standardized observed variables by the following matrix:
\begin{align}
    \mathbf{m} =&
    \bm{\mu}_z + (\Sigma_z^{-1} + C \hat{W}^T \hat{W} C)^{-1} \hat{W}^T C \Psi^{-1/2} (\mathbf{x} -\bm{\mu}_x), \notag \\
    =& 
    \bm{\mu}_z + (\Sigma_z^{-1} + C \hat{W}^T \hat{W} C)^{-1} \hat{W}^T C (1 + C^2)^{1/2} \mathrm{diag}(\Sigma_x)^{-1/2} ( \mathbf{x} -\bm{\mu}_x), \notag \\
    \rightarrow & \bm{\mu}_z + (\hat{W}^T \hat{W})^{-1} \hat{W}^T \mathrm{diag}(\Sigma_x)^{-1/2} (\mathbf{x} - \bm{\mu}_x).
\end{align}
In other words, the factor scores can be interpreted as the coordinate components of the standardized observed variables $\mathrm{diag}(\Sigma_x)^{-1/2}(\mathbf{x} - \bm{\mu}_x)$ in the basis vectors, the column vectors of matrix $\hat{W}$.

\section{Numerical validation} \label{sec:numerical_validation}
In this section, we numerically validate the proposed factor analysis using real and artificial datasets.

Here, we discuss a method of parameter estimation from observed data $\mathcal{D} \equiv \{(\mathbf{x}_1, \mathbf{y}_1), (\mathbf{x}_2, \mathbf{y}_2), \dots, (\mathbf{x}_N, \mathbf{y}_N) \}$.
In this paper, we adopt parameter estimation based on maximum likelihood estimation.
In the proposed model, the log-likelihood function $l(\bm{\theta})$ is analytically expressed as
\begin{align}
    l(\bm{\theta}) =& \sum_{i=1}^N \log p(\mathbf{x}_i, \mathbf{y}_i \mid \bm{\theta}), \hspace{0.5cm}
    \bm{\theta} = (\bm{\mu}_x, \Psi, W, \mathbf{b}, G).
\end{align}
Since the optimization of the likelihood function is a nonlinear optimization problem with numerous local maxima, it is difficult to obtain the maximum likelihood estimator.
Therefore, in this paper, we randomly generated initial parameters more than 100 times and numerically solved the optimization problem for each initial parameter, and then adopted the parameter with the highest likelihood as the optimal solution.

\subsection{Application to real datasets}
In this subsection, we numerically validate the proposed factor analysis and the norm constraint prescription using a publicly available real dataset, HIV drug resistance data.
We analyzed the mutation of amino acid sequences of Human Immunodeficiency Virus (HIV) type-1.
The dataset was obtained from the HIV drug resistance database~\cite{hiv}.
Details of the database and related datasets can be found in \cite{Rhee2003}.
When an antiretroviral drug is dosed on a patient, the virus becomes resistant to the drug over time by mutating its genes.
This mutation has been observed to be highly cooperative, with each residue in the amino acid sequence mutating not simply stochastically~\cite{Ohtaka2003}.
Although the molecular mechanism of drug resistance has not yet been elucidated, it is expected that the relationship between the correlation pattern of mutations and drug resistance will provide clues to the molecular mechanism of drug resistance.
We focused on viral resistance to protease inhibitors.
The data for analysis consists of mutational information on residues of amino acid sequences from position 1 to 99 in protease of viruses isolated from plasma of HIV-1 infected patients, represented by \texttt{P1} to \texttt{P99}, and in vitro susceptibility to various protease inhibitors such as Nelfinavir.
As a preprocessing step, the residues of amino acid sequences were encoded to $1$ if any mutation, such as insertion, deletion, or substitution from the consensus wild-type amino acid sequence, is present, and encoded to $0$ if there is no mutation from the consensus sequence.
In other words, information on the type of mutation was ignored.
The variables of the drug susceptibility were transformed using a natural logarithmic function to approximately satisfy the normality assumption of the proposed factor analysis.

To reduce the model complexity relative to the sample size, we chose three continuous variables and eight binary variables out of the variables of the drug resistance and the residues of amino acid sequences.
The three continuous variables are the susceptibility to the protease inhibitors: Indinavir (IDV), Nelfinavir (NFV) and Squinavir (SQV).
Among them, Indinavir is used for coloring the scatter plot of the factor scores.
These variables were selected in terms of having few missing values.
The eight binary variables are the residues of amino acid sequences at the positions $\texttt{P63}$, $\texttt{P10}$, $\texttt{P71}$, $\texttt{P36}$, $\texttt{P46}$, $\texttt{P54}$, $\texttt{P93}$ and $\texttt{P62}$, which were selected from the 99 residues in descending order of mutation rate.
The proposed model naturally handles missing values since it is based on probability distributions.
However, we used only the data without missing values in the mutations and the drug susceptibility to reduce the complexity of programming implementation.
We also excluded data with multiple isolates for the same patient.
The sample size of the dataset after data cleansing was $N=1906$.
Figure~\ref{fig:hiv_descriptive} shows the descriptive statistics of the HIV data.
We observe that there are strong correlations among mutations in the amino acid residues and drug resistance.

\begin{figure}[htbp]
    \centering \hspace*{-0cm}
    \includegraphics[scale=1, pagebox=cropbox, viewport= 0 0 448.959 193.934, clip]{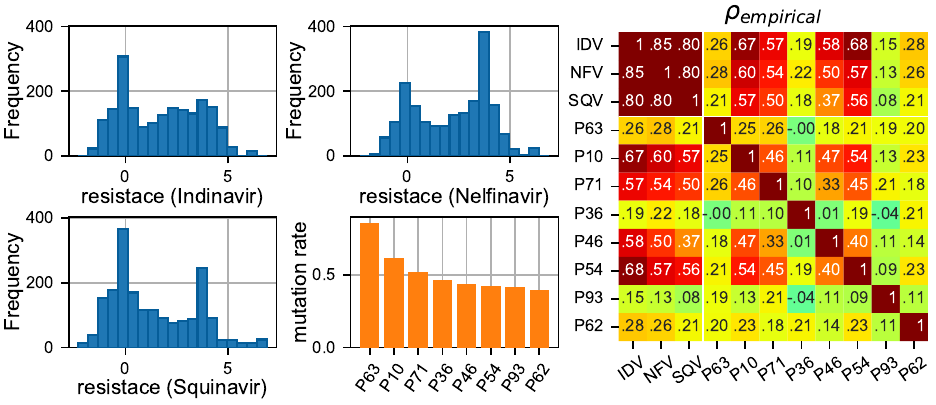}
    \caption{
    Descriptive statistics of the HIV drug resistance data: histogram of the continuous variables, means of the binary variables and correlation matrix, where the mean and correlation for binary variables are defined using dummy variables.
    }
    \label{fig:hiv_descriptive}
\end{figure}

We performed the proposed factor analysis on the mutation data of the amino acid sequences in HIV-1 protease.
Model parameters were estimated by maximum likelihood estimation.
The Bayesian information criterion (BIC) was used to determine the number of latent dimensions~\cite{BIC}.
We also performed factor analysis through the method of quantification for comparison.
The BIC as a function of the number of latent dimensions is shown in Fig.~\ref{fig:bic}.
The number of latent dimensions selected for the proposed model was four.
In the following analysis, the number of latent dimensions for the factor analysis with the method of quantification were set to the same value as that of the proposed factor analysis for comparison and the norm constraint prescription was applied as well.

\begin{figure}[htbp]
    \centering
    \includegraphics[scale=1, pagebox=cropbox, viewport= 0 0 348.710 139.367, clip]{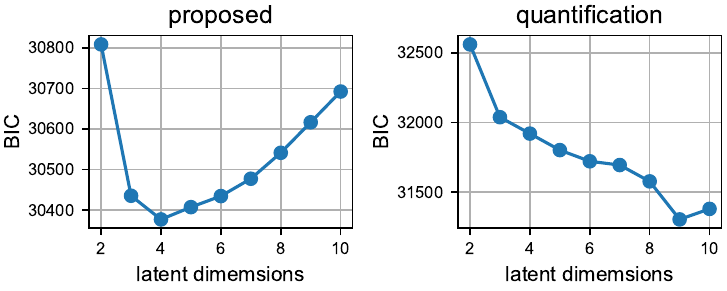}
    \caption{
    The BIC as a function of the number of latent dimensions for the proposed factor analysis (Left) and factor analysis with quantification (Right).
    }
    \label{fig:bic}
\end{figure}

Factor analysis allows us to visualize the relationship between data points and features, which is known as a biplot.
In the biplot, the factor scores of each data point $\mathbf{m}_i, \; (i=1,2,\dots, N)$, Eq.~(\ref{eq:factor_score}), are depicted as a scatter plot in Euclidean space, and the dimensionless factor loading vectors, which are defined by $c \, \hat{\mathbf{w}}_j$ and $\mathbf{g}_j = c \, \hat{\mathbf{g}}_j$, are depicted as arrows.
When the observed variables consist exclusively of binary variables, the factor score consists of $2^q$ possible combinations of the factor loading vectors $\mathbf{g}_j$.
The Euclidean distance between two data points represents the similarity between them.
The inner product of the arrows of two features represents the similarity between them.
Our norm constraint prescription is also convenient in comparing features with each other in the biplot, since the arrows of all features are the same length.
The inner product of data point $\mathbf{m}_i$ and the dimensionless factor loading vector $c \, \hat{\mathbf{w}}_j$ and $\mathbf{g}_j$ means that the corresponding feature is relatively larger or more likely to occur than mean value at that data point.

The numerical analysis demonstrated that both models exactly reproduce the empirical mean of the data.
Figure~\ref{fig:biplot_hiv} shows comparisons of the correlation matrix of the HIV data with those reproduced by the both models and biplots of factor analysis.
We see that our model successfully reproduces the empirical correlation.
In the biplots, the first and the second latent dimensions are displayed, and the percentages in the axis labels represent the contribution ratio, Eq.~(\ref{eq:contribution_ratio}).
From the biplot, we see that the first latent dimension can be interpreted as the resistance to the protease inhibitor, and the second and subsequent latent dimensions appear to be irrelevant to drug resistance.
We also see qualitative similarities between the proposed factor analysis and factor analysis through the method of quantification.
However, these two methods showed quantitative differences.
In particular, the contribution ratio of the latent dimension was different between these methods.

\begin{figure}[htbp]
    \centering
    \includegraphics[scale=1, pagebox=cropbox, viewport= 0 0 389.388 392.939, clip]{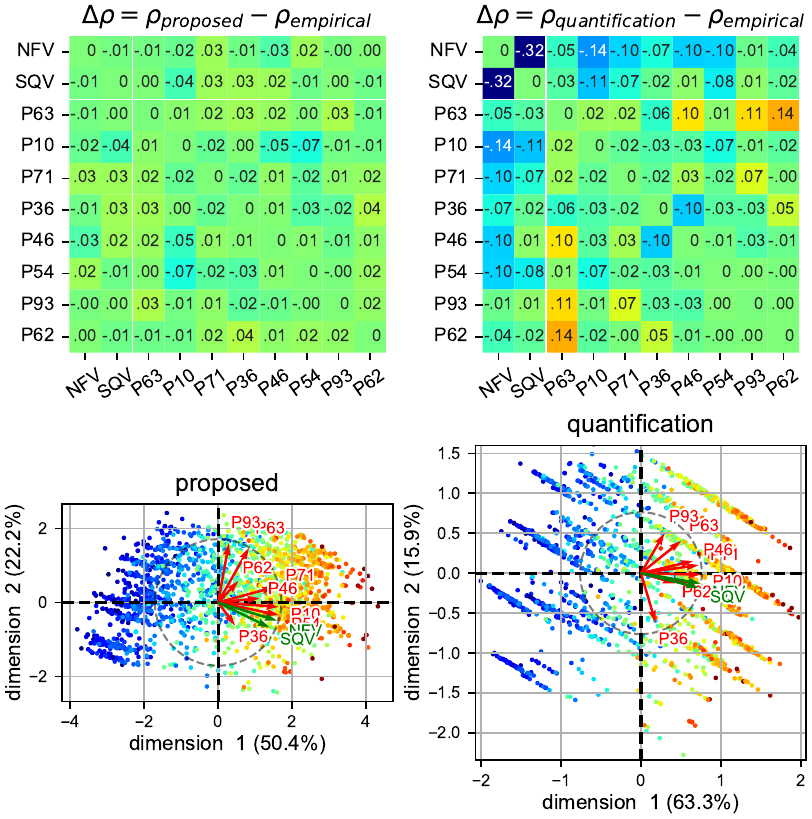}
    \caption{
    Correlation matrices reproduced by the models and biplots of HIV-1 protease mutation data.
	The markers in the scatter plot are colored from blue to red according to weak to strong resistance to the protease inhibitor (Indinavir).
    The dashed circle represents the maximum possible length of the factor loading vectors.
	The arrow lengths of the factor loading vectors are scaled to match the points of the factor scores $\mathbf{m}$.
    }
    \label{fig:biplot_hiv}
\end{figure}

\subsection{Validation with artificially generated datasets}
The analysis of real data in the previous section did not allow us to examine in detail how the proposed method and quantification method differ.
Therefore, in this section, we first quantitatively compare the representational capability of the proposed method and the quantification method using artificial data.
We artificially generated continuous and binary variables and checked whether the model can accurately describe the correlation of the generated data.

We used the following procedure to generate artificial data for mixed continuous and binary variables using random numbers.
To generate artificial data for mixed-type variables, we first generated correlated continuous variables from a multivariate standard normal distribution.
Binary variables were then obtained by converting the multivariate continuous variables into dichotomous variables through thresholding.
The parameters of the model, used to create the data, were generated using random numbers as follows.
We generated the covariance matrix of the multivariate standard normal distribution using spectral decomposition.
First, the eigenvalues of the positive definite symmetric matrix were generated using a Gamma distribution $\mathrm{Gamma}(\alpha=1, \beta=1)$.
Next, we generated vectors with elements following the standard normal distribution and orthonormalized them using the Gram-Schmidt orthonormalization to obtain the orthonormal eigenvectors for each eigenvalue.
Using these eigenvalues and eigenvectors, we formed a positive definite matrix through spectral decomposition.
This positive definite symmetric matrix was then standardized to obtain the covariance matrix of the multivariate standard normal distribution.
The continuous variables generated with this covariance matrix were dichotomized by applying thresholds based on specific percentiles of the standard normal distribution.
The thresholds used for dichotomization were generated from the continuous uniform distribution ranging from zero to one, $\mathrm{Unif}(0, 1)$.

We generated 500 sets of artificial data, each with a sample size of 1000.
In that data, there are five continuous variables and five binary variables, making a 10-dimensional dataset.
We excluded the artificial data in which any of the binary variables had an average of either zero or one, since correlation cannot be defined for such data.
We also excluded data in which the maximum correlation between continuous and binary variables and between binary variables did not exceed 0.5 simultaneously.
The model parameters were estimated using maximum likelihood estimation.
The models of the proposed and quantification methods were then fitted to the artificially generated data by maximum likelihood estimation, and the reproducibilities of correlations by the models were examined.

Figures~\ref{fig:reproducibility_proposed} and~\ref{fig:reproducibility_quantification} show scatter plots comparing the empirical correlations of the synthetic data and correlations reproduced by the models.
In these figures, the coefficient of determination is defined by
\begin{align}
    R^2 = 1 - \frac{\sum_{i > j} (r_{ij} - \hat{r}_{ij})^2}{\sum_{i > j} (r_{ij} - \bar{r}_{ij})^2}, \hspace{0.5cm} i, j \in \{1, 2, \dots, p_x + q_y \},
\end{align}
where $r_{ij}$ and $\hat{r}_{ij}$ are empirical correlation and correlation reproduced by the model, respectively, and $\bar{r}_{ij}$ is a mean value of the empirical correlations.
In comparison with quantification methods, it is observed that the proposed method achieves better reproducibility of correlations than the quantification methods.
We observe the disadvantage of the quantification method that it is more difficult to reproduce the correlations involving binary variables compared to those involving only continuous variables.
This drawback arises when the representational capacity is insufficient due to the small number of the latent dimensions, and when the variance of the binary variable is small, which corresponds to cases where the mean of the binary variable is close to either 0 or 1.
Specifically, the quantification method tends to overestimate correlations involving binary variables when the variance of those binary variables is small. 
This suggests that the discreteness of binary variables becomes particularly pronounced when the variance of binary variables is small, and at the same time, it indicates that the assumptions underlying the quantification method become invalid.

\begin{figure}[htbp]
    \centering
    \includegraphics[scale=1, pagebox=cropbox, viewport= 0 0 389.356 394.834, clip]{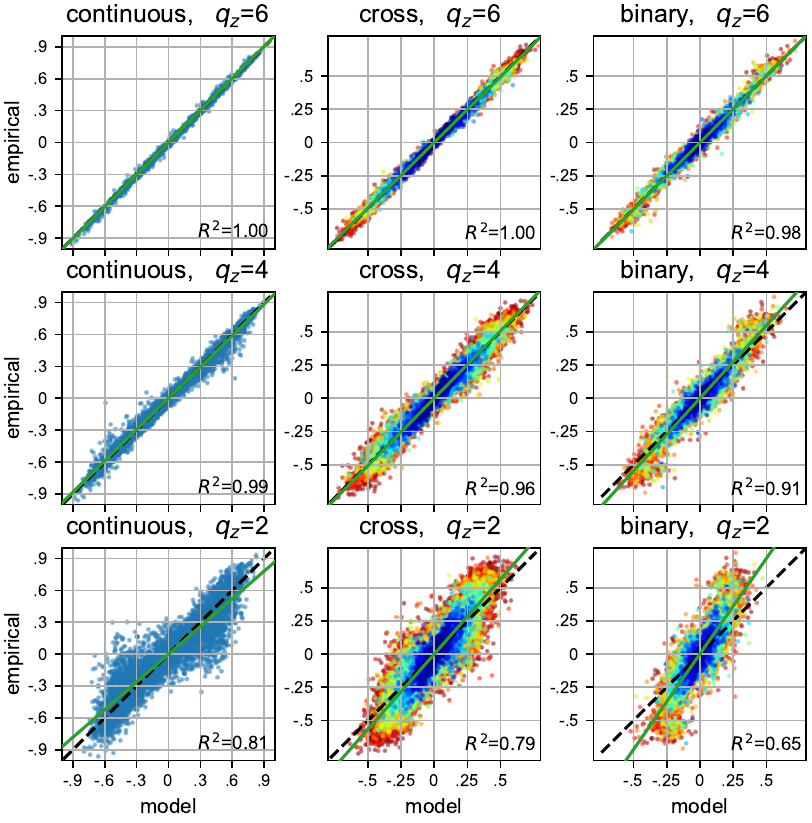}
    \caption{
    Scatter plots of empirical correlations and correlations reproduced by the proposed factor analysis as a function of the number of latent dimensions.
	The green solid lines represent the regression lines obtained by regressing the empirical correlation using the correlation reproduced by the model.
	$R^2$ indicates the coefficient of determination between empirical correlations and correlations reproduced by the models.
	In correlations involving binary variables, the point colors change from blue to red as the variance of the binary variable increases.
    In correlations between binary variables, blue points represent variables for which the mean of one of the binary variables is less than 0.1 or greater than 0.9.
    In correlations between continuous and binary variables, blue points represent variables for which the mean of the binary variable is less than 0.1 or greater than 0.9.
    }
    \label{fig:reproducibility_proposed}
\end{figure}

\begin{figure}[htbp]
    \centering
    \includegraphics[scale=1, pagebox=cropbox, viewport= 0 0 389.356 394.834, clip]{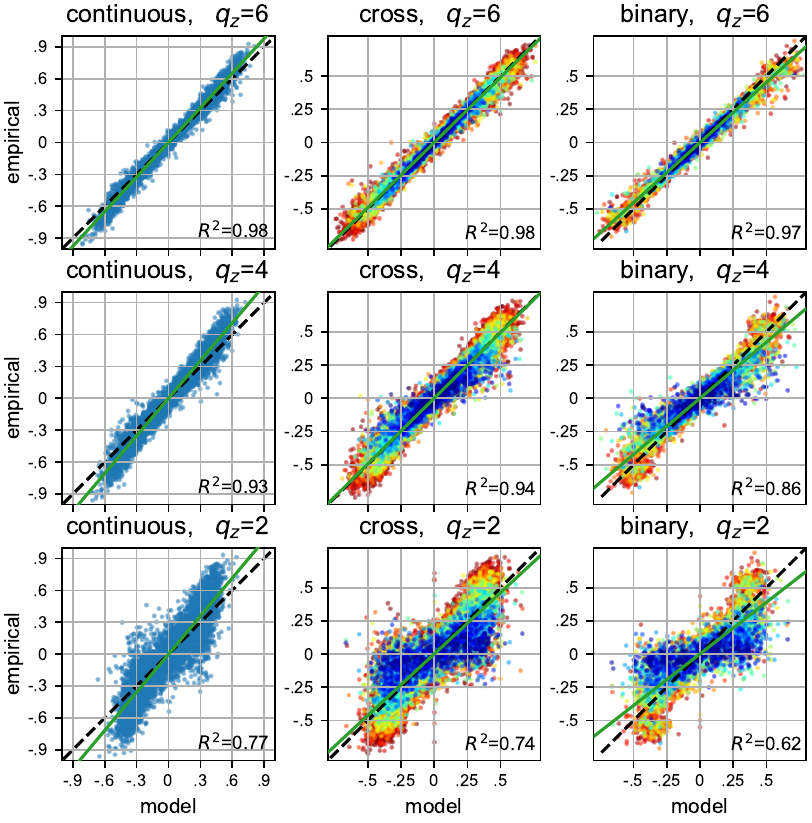}
    \caption{
	Scatter plots of empirical correlations and correlations reproduced by the factor analysis with the method of quantification as a function of the number of latent dimensions.
    The green solid lines represent the regression lines obtained by regressing the empirical correlation using the correlation reproduced by the model.
	$R^2$ indicates the coefficient of determination between empirical correlations and correlations reproduced by the models.
	In correlations involving binary variables, the point colors change from blue to red as the variance of the binary variable increases.
    In correlations between binary variables, blue points represent variables for which the mean of one of the binary variables is less than 0.1 or greater than 0.9.
    In correlations between continuous and binary variables, blue points represent variables for which the mean of the binary variable is less than 0.1 or greater than 0.9.
    }
    \label{fig:reproducibility_quantification}
\end{figure}

However, there is a way to avoid the drawback of the quantification.
This way involves fitting the data using a model with a sufficiently large number of latent dimensions and then removing the information from unnecessary latent dimensions.
We first parametrize the covariance matrix as
\begin{align}
	\Sigma_x =& \bigl[ I - \mathrm{diag}(\bar{W} \bar{W}^T ) \bigr] \mathrm{diag}(\Sigma_x) + \mathrm{diag}(\Sigma_x)^{1/2} \bar{W} \bar{W}^T \mathrm{diag}(\Sigma_x)^{1/2}, \\
    \bar{W} \equiv & \sqrt{\frac{c^2}{1 + c^2}} \hat{W},
\end{align}
and fit the model with a sufficiently large number of latent dimensions $p_z = p_x + q$.
Then, we drop the information from the latent dimensions, the columns of the factor loading matrix $\bar{W}$, so that the diagonal elements of the covariance matrix $\Sigma_x$ remain unchanged.
Figure~\ref{fig:reproducibility_quantification_max_reduce} shows scatter plots of the empirical correlations and correlations reproduced by the quantified method, where the model were fit with $q_z=9$ and then reduced the latent dimensions.
This model can reproduce correlations as well as the proposed model, however the characteristics of attempting to overestimate correlations involving binary variables when the variance of the binary variable is small still remains.

\begin{figure}[htbp]
    \centering
    \includegraphics[scale=1, pagebox=cropbox, viewport= 0 0 389.356 394.834, clip]{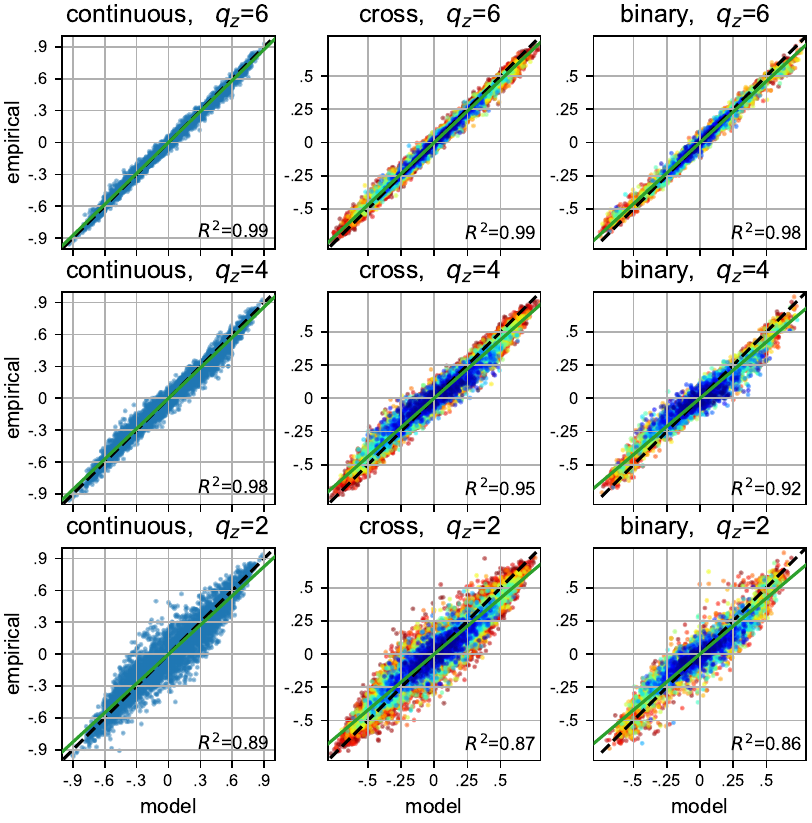}
    \caption{
	Scatter plots of empirical correlations and correlations reproduced by quantification factor analysis, in which the latent dimensions are reduced, as a function of the number of latent dimensions.
	The green solid lines represent the regression lines obtained by regressing the empirical correlation using the correlation reproduced by the model.
	$R^2$ indicates the coefficient of determination between empirical correlations and correlations reproduced by the models.
	In correlations involving binary variables, the point colors change from blue to red as the variance of the binary variable increases.
    In correlations between binary variables, blue points represent variables for which the mean of one of the binary variables is less than 0.1 or greater than 0.9.
    In correlations between continuous and binary variables, blue points represent variables for which the mean of the binary variable is less than 0.1 or greater than 0.9.
    }
    \label{fig:reproducibility_quantification_max_reduce}
\end{figure}

\subsection{Sampling distribution of the maximum likelihood estimates}
In this subsection, we show that the proposed factor analysis reproduces the underlying latent structure when the sample size is large.
We also investigate the sampling distribution of the maximum likelihood estimates.

To this end, we first adopt the model and factor scores obtained by estimating the artificial data generated in the previous subsection as the ground truth.
The sample size of the dataset is $N=1000$, and the number of latent dimensions is six.
We generated synthetic observed data $(\mathbf{x}_i, \mathbf{y}_i), \hspace{0.2cm} i=1,2,\dots,1000$ using the conditional distribution $p(\mathbf{x}, \mathbf{y} \mid \mathbf{z})$ based on the ground truth model and factor scores $\mathbf{m}_i, \hspace{0.2cm} i=1,2,\dots,1000$.
We then applied the proposed factor analysis to the synthetic data to examine whether it could successfully reproduce the ground truth.
Figure~\ref{fig:artificial_ground_truth} shows the biplots for the ground truth model and proposed model, where the feature indices from zero to four correspond to binary variables and those from five to nine correspond to continuous variables.
It can be seen that our method reproduces the ground truth.

\begin{figure}[htbp]
    \centering
    \includegraphics[scale=1, pagebox=cropbox, viewport= 0 0 391.388 175.533, clip]{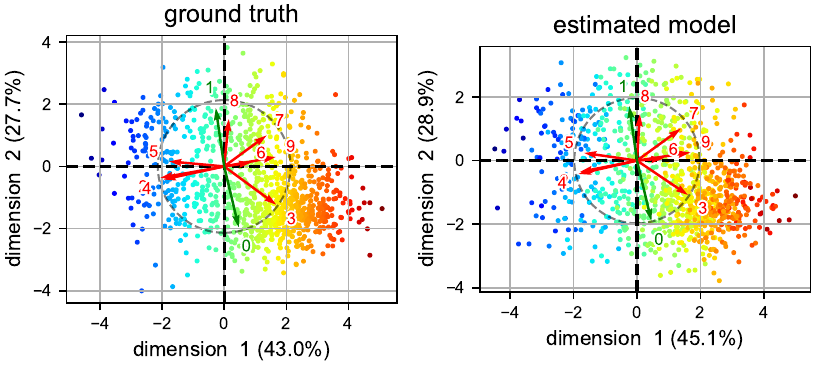}
    \caption{
    Correlation matrix and biplot of the ground truth model (Left) and correlation matrix and biplot from the model (Right).
	The markers in the scatter plot are colored from blue to red according to small to large factor scores of the first principal axis $z_1$ in the ground truth data.
    The dashed circle represents the maximum possible length of the factor loading vectors.
	The arrow lengths of the factor loading vectors are scaled to match the points of the factor scores $\mathbf{m}$.
    }
    \label{fig:artificial_ground_truth}
\end{figure}

Next, we investigate the sampling distribution of the maximum likelihood estimates and examine the consistency of the model parameters.
To investigate the sampling distribution, we generated synthetic data from the joint distribution $p(\mathbf{x}, \mathbf{y})$, Eq.~(\ref{eq:fa_induced}), with the ground truth model parameter.
We generated 1000 sets of synthetic data with sample sizes of 1000, 3000, and 9000, respectively.
We then fitted the model to these datasets using maximum likelihood estimation.
Figure~\ref{fig:sampling_distribution} shows the sampling distribution of the model parameters.
We empirically observed that as the sample size increases, the estimates exhibit consistency and asymptotic normality.

\begin{figure}[htbp]
    \centering \hspace*{-0cm}
    \includegraphics[scale=1, pagebox=cropbox, viewport= 0 0 415.027 240.930, clip]{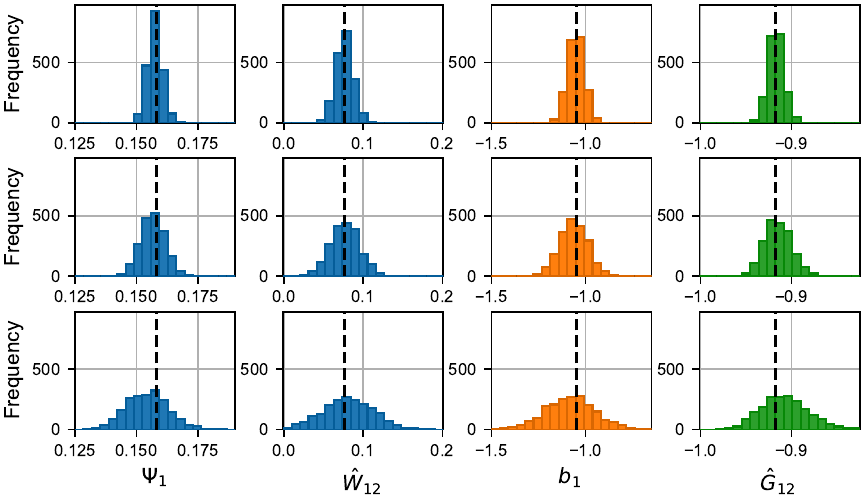}
    \caption{
        Sampling distribution of the maximum likelihood estimates for model parameters.
        The true parameters are indicated by the dashed lines.
        The sample sizes are 1000 for the lower row figures, 3000 for the middle row figures, and 9000 for the upper row figures.
    }
    \label{fig:sampling_distribution}
\end{figure}

\section{Conclusions} \label{sec:conclusion}
We developed factor analysis for mixed continuous and binary random variables.
The proposed factor analysis amounts to introducing a mixture of Gaussian distributions with specified parameters for a prior distribution over the latent variables.
The introduction of this mixture prior enables closed-form marginalization over the latent variables and allows the induced distribution to be expressed analytically.
Hence, the proposed method can easily perform parameter estimation even when the number of latent dimensions is large, which causes a difficulties in parameter estimation for the existing methods of binary Factor Analysis and multidimensional IRT.
On the other hand, the proposed method requires summing over all possible states when calculating the partition function.
Since the number of all possible states of binary variables exponentially increases as the number of binary variables increases, the issue of computational complexity arises when the number of binary variables is large.
We also discussed the instability of model parameters associated with maximum likelihood estimation in the proposed factor analysis, which is known as improper solutions or Heywood cases.
We proposed a prescription to fix this instability, which imposes the constraint that the row vectors of the factor loading matrix have the same norm for all features.
This is equivalent to imposing the communalities to have the same value for all features.
Note that all results in this paper were obtained by applying the norm constraint prescription.

By analyzing a real dataset, we numerically confirmed that this norm constraint prescription works well and prevents instability of estimated model parameters.
We also compared the proposed method with factor analysis with the quantification method on synthetic data.
The proposed method exhibited better reproducibility of correlations involving binary variables than the quantification method.
We identified a drawback of the quantification method through numerical experiments.
We found that the quantification method tends to overestimate the correlations involving binary variables with small variance when the number of latent dimensions is small.
In contrast, this tendency was not observed in the proposed method.

The proposed factor analysis can be used for compressing and visualizing data that include binary variables, owing to the uniqueness of the estimated model parameters ensured by the model's identifiability and the absence of improper solutions.
Furthermore, the method can also be used as a preprocessing step for dimensionality reduction and denoising of feature variables, before applying non-probabilistic machine learning methods such as support vector machines.
Although PCA and conventional factor analysis have become standard tools in data analysis for dimensionality reduction and preprocessing tasks, the justification for applying these methods to datasets containing binary features is debatable because they are inherently designed for continuous features.
The proposed factor analysis is expected to be a principled alternative to these conventional techniques when dealing with mixed continuous and binary variables.

\appendix

\section{Derivation of the proposed factor analysis} \label{sec:appendix_fa}
In this appendix, we first provide a summary of the multivariate distribution that describe correlations between continuous and binary variables, introduced in Ref.~\cite{Arai2022}.
Then, we derive the proposed factor analysis.

\subsection{Probability distribution for mixed continuous and binary variables}
We denote continuous and binary observed variables by $\mathbf{x}$ and $\mathbf{y}$.
Each variable is a column vector and its dimensions are $p$, $q$, respectively.
The distribution is parametrized by $(\mu, \Sigma, \Lambda, G)$.
The parameters $\mu$ and $\Sigma$ are a column vector and square matrix with dimension $p$ representing the mean and covariance of the continuous variables, respectively.
The parameter $\Lambda$ is a $q \times q$ matrix for the Grassmann distribution.
$\Lambda - I$ must be a $P_0$-matrix, i.e., all principal minors are nonnegative, where $I$ is an identity matrix.
The parameter $G$ is a $q \times p$ matrix that describes the correlation between continuous and binary variables.

To express the joint distributions, we define the notation of index labels.
We denote the set of whole indices of continuous and binary variables as $I\equiv \{1,2,\dots,p\}$ and $R \equiv \{1,2,\dots, q\}$, respectively.
The set of whole indices for binary variables is divided into two parts with subscripts of $1$ or $0$, $R = (R_1, R_0)$, for variables that take a value of $1$ or $0$, respectively.
Then, we write a subvector of the binary variables taking a value of $1$ and $0$ as $\mathbf{y}_{R_1}=\bm{1}$ and $\mathbf{y}_{R_0}=\bm{0}$, respectively.
Using these notations, the joint distribution can be calculated as follows:
\begin{align}
    p(\mathbf{x},\mathbf{y}) =& p(\mathbf{x}, \mathbf{y}_{R_1}=\bm{1}, \mathbf{y}_{R_0}=\bm{0}) , \notag \\
    =& \frac{1}{Z} \det (\Lambda_{R_0 R_0} -I) \, e^{\frac{1}{2} \mathbf{g}_{R_0}^T \Sigma \mathbf{g}_{R_0}} e^{-\frac{1}{2}(\mathbf{x}-\bm{\mu} + \Sigma \mathbf{g}_{R_0})^T \Sigma^{-1}(\mathbf{x}-\bm{\mu} + \Sigma \mathbf{g}_{R_0})}, \\
    G =&  \begin{bmatrix} \mathbf{g}_1, &  \mathbf{g}_2, & \dots, & \mathbf{g}_q \end{bmatrix}^T, \notag \\
    Z =& 
    (2\pi)^{p/2} \det \Sigma^{1/2} \sum_{R_0 \subseteq R} \det (\Lambda_{R_0 R_0}-I) e^{\frac{1}{2} \mathbf{g}_{R_0}^T \Sigma \mathbf{g}_{R_0}},
\end{align}
where $\Lambda_{R_0 R_0}$ is a submatrix of $\Lambda$, with its rows and columns consisting of the index set $R_0$, and we have defined a column vector $\mathbf{g}_{R_0} \equiv \sum_{s \in R_0} \mathbf{g}_s$.
The summation $\sum_{R_0 \subseteq R}$ runs over all possible principal minors.
That is, the partition function requires summing over all possible states for binary variables.
The above expression can be viewed as a parametrized location model~\cite{Olkin1961}.

To express the marginal and conditional distributions, we define the notation of index labels.
We denote the index label of a subset of indices as $J \subseteq I$. 
The subvector consisting of the subset of indices $J$ is represented by $\mathbf{x}_J$.
We divide the set of whole indices of continuous and binary variables $I$ and $R$ into three subset parts; $I=(J, L, K)$ and $R=(S, U, T)$, where the index labels for a set of indices $L$ and $U$ are introduced to handle missing values.
Hence, the index labels $L$ and $U$ may be understood as the initial letters of ``Latent'' and ``Unobserved'', respectively.
The number of elements in the set of indices is represented by $p_J$, $p_L$, $p_K$ and $q_S$, $q_U$, $q_T$, these satisfy $p_J + p_L + p_K = p$ and $q_S + q_U + q_T = q$.
Then, the vectors $\mathbf{x}$ and $\mathbf{y}$ can be partitioned into subvectors as $\mathbf{x}=\mathbf{x}_I = (\mathbf{x}_J, \mathbf{x}_L, \mathbf{x}_K)$ and $\mathbf{y}=\mathbf{y}_R=(\mathbf{y}_S, \mathbf{y}_U, \mathbf{y}_T)$, respectively.
Again, a subset of indices for binary variables, e.g., $S$, is divided into two parts; we write the index label for variables that take a value of $y_s=1$ and $y_s=0, \;\; (s \in S)$ as $S_1$ and $S_0$, and denote these variables as $\mathbf{y}_{S_1}$ and $\mathbf{y}_{S_0}$, respectively.
The union of the index label $J$ and $K$ is denoted as $J + K \equiv J \cup K$.

By marginalizing the variables $\mathbf{x}_{J+L}$ and $\mathbf{y}_{S+U}$, the marginal distribution can be calculated as follows:
\begin{align}
    p(\mathbf{x}_K, \mathbf{y}_T) =& p(\mathbf{x}_K, \mathbf{y}_{T_1}=\bm{1}, \mathbf{y}_{T_0}=\bm{0}), \notag \\
    =&
	\frac{1}{Z} (2\pi)^{p_{J+L}/2} \det \Sigma_{(J+L)|K}^{1/2} \notag \\
	& \hspace{0cm} \sum_{S_0 \subseteq S, \; U_0 \subseteq U} \det(\Lambda_{R_0 R_0} - I)  e^{\frac{1}{2} \mathbf{g}_{R_0}^T \Sigma \mathbf{g}_{R_0}}
    e^{-\frac{1}{2} (\mathbf{x}_K -\bm{\mu}_K + \Sigma_{KI} \mathbf{g}_{I R_0})^T \Sigma_{KK}^{-1} (\mathbf{x}_K -\bm{\mu}_K + \Sigma_{KI} \mathbf{g}_{I R_0})},
\end{align}
where
\begin{align}
	\Sigma_{(J+L) | K} \equiv \Sigma_{(J+L)(J+L)} - \Sigma_{(J+L) K} \Sigma_{KK}^{-1} \Sigma_{K (J+L)} 
\end{align}
is the Schur complement of the matrix $\Sigma$ with respect to $\Sigma_{KK}$, $\Sigma_{KK}^{-1} \equiv (\Sigma_{KK} )^{-1}$ is the inverse of the submatrix, and we have defined the division of the vector $\mathbf{g}_{R_0}^T$ as $\mathbf{g}_{R_0}^T = \mathbf{g}_{R_0 I}^T = (\mathbf{g}_{R_0 J}^T, \mathbf{g}_{R_0 L}^T, \mathbf{g}_{R_0 K}^T)$.
The conditional distribution with marginalized variables can be derived by dividing marginal distributions:
\begin{align}
    & p(\mathbf{x}_J, \mathbf{y}_{S_1}=\bm{1}, \mathbf{y}_{S_0}=\bm{0} |\mathbf{x}_K, \mathbf{y}_{T_1}=\bm{1}, \mathbf{y}_{T_0}=\bm{0}) 
    = \frac{p(\mathbf{x}_{J+K}, \mathbf{y}_{S_1+T_1}=\bm{1}, \mathbf{y}_{S_0+T_0}=\bm{0})}{p(\mathbf{x}_K, \mathbf{y}_{T_1}=\bm{1}, \mathbf{y}_{T_0}=\bm{0})}, \notag \\
    =& 
    \frac{ \sum_{U_0 \subseteq U} \det(\Lambda_{R_0 R_0}-I)
    e^{\frac{1}{2} \mathbf{g}^T_{R_0 (J+L)} \Sigma_{(J+L)|K} \mathbf{g}_{(J+L) R_0}}
    e^{-\mathbf{g}_{(S_0+U_0) I}^T \Sigma_{IK} \Sigma_{KK}^{-1} (\mathbf{x}_K -\bm{\mu}_K)}}
    {\sum_{S_{0}' \subseteq S, \; U_{0}' \subseteq U} \det(\Lambda_{R_{0}' R_{0}'} - I)
    e^{\frac{1}{2} \mathbf{g}^T_{R_0' (J+L)} \Sigma_{(J+L)|K} \mathbf{g}_{(J+L) R_0'}}
    e^{-\mathbf{g}_{(S_{0}' + U_{0}') K}^T \Sigma_{IK} \Sigma_{KK}^{-1} (\mathbf{x}_K -\bm{\mu}_K)}  } \notag \\
    & \hspace{-0cm}
    \frac{1}{(2\pi)^{p_J/2} \det\Sigma_{J|K}^{1/2}} 
	\exp \biggl\{ -\frac{1}{2}(\mathbf{x}_J - \bm{\mu}_J + \Sigma_{JI} \mathbf{g}_{IR_0} - \Sigma_{JK} \Sigma_{KK}^{-1} (\mathbf{x}_K - \bm{\mu}_K + \Sigma_{KI} \mathbf{g}_{IR_0}))^T \Sigma_{J|K}^{-1} \notag \\
	& \hspace{4.75cm} (\mathbf{x}_J - \bm{\mu}_J + \Sigma_{JI} \mathbf{g}_{IR_0} - \Sigma_{JK} \Sigma_{KK}^{-1} (\mathbf{x}_K - \bm{\mu}_K + \Sigma_{KI} \mathbf{g}_{IR_0})) \biggr\}.
\end{align}

When we define a $q$-dimensional bit vector $\bm{1}_{R_0}$ whose elements take a value of $0$ or $1$,
\begin{align}
    [\bm{1}_{R_0}]_s \equiv
    \begin{cases} 1, \hspace{0.3cm} \text{if} \; \; s \in R_0 \\ 
        0, \hspace{0.3cm} \text{if} \; \; s \in R_1
    \end{cases}, \; \; (s=1, 2, \dots, q),
    \label{eq:bit_vector_appendix}
\end{align}
the column vector $\mathbf{g}_{R_0}$ can also be expressed as
\begin{align}
    \mathbf{g}_{R_0} \equiv \sum_{s \in R_0} \mathbf{g}_s = G^T \bm{1}_{R_0}
    = \begin{bmatrix} \mathbf{g}_1, \;\; \mathbf{g}_2, \;\; \dots, \;\; \mathbf{g}_q \end{bmatrix} \bm{1}_{R_0}.
\end{align}
Using the bit vector defined in the above equation, we obtain the following expressions for the joint, marginal, and conditional distribution:
\begin{align}
    p(\mathbf{x}, \mathbf{y}=\bm{1}_{R_1}) 
    =& \pi_{R_0}(\Sigma) \, \mathcal{N}( \mathbf{x} \mid \bm{\mu} - \Sigma G^T \bm{1}_{R_0}, \Sigma ), \\
    \pi_{R_0}(\Sigma)
    \equiv&
    \frac{\det \bigl[(\Lambda -I)_{R_0 R_0} \bigr]  e^{\frac{1}{2} \bm{1}_{R_0}^T G \Sigma G^T \bm{1}_{R_0}}}{\sum_{R_0' \subseteq R} \det \bigl[(\Lambda-I)_{R_0' R_0'}\bigr] e^{\frac{1}{2} \bm{1}_{R_0'}^T G \Sigma G^T \bm{1}_{R_0'}}},
\end{align}
\begin{align}
    p(\mathbf{x}_K, \mathbf{y}_T) 
    =&
    \sum_{S_0 + U_0 \subseteq R \setminus T} \pi_{R_0}(\Sigma) \,
    \mathcal{N}( \mathbf{x}_K \mid \bm{\mu}_K - \Sigma_{KI} G^T \bm{1}_{R_0}, \Sigma_{KK} ),
\end{align}
\begin{align}
    & p(\mathbf{x}_J, \mathbf{y}_S |\mathbf{x}_K, \mathbf{y}_T)  \notag \\
    =&
    \frac{\sum_{U_0 \subseteq R\setminus(S+T)} \det \bigl[ (\Lambda-I)_{R_0 R_0} \bigr]
    e^{\frac{1}{2} \bm{1}_{R_0}^T G_{R(J+L)} \Sigma_{(J+L)|K} G_{(J+L)R}^T \bm{1}_{R_0}}
    e^{-(\bm{1}_{S_0} + \bm{1}_{U_0})^T G \Sigma_{IK} \Sigma_{KK}^{-1} (\mathbf{x}_K - \bm{\mu}_K)}}
    {\sum_{S_0' + U_0' \subseteq R \setminus T} \det\bigl[ (\Lambda-I)_{R_0' R_0'} \bigr]
    e^{\frac{1}{2} \bm{1}_{R_0'}^T G_{R(J+L)} \Sigma_{(J+L)|K} G_{(J+L)R}^T \bm{1}_{R_0'}}
    e^{-(\bm{1}_{S_0'} + \bm{1}_{U_0'})^T G \Sigma_{IK} \Sigma_{KK}^{-1} (\mathbf{x}_K - \bm{\mu}_K)}} \notag \\
    &
    \mathcal{N}\bigl(\mathbf{x}_J \mid \bm{\mu}_J - \Sigma_{JI} G^T \bm{1}_{R_0} + \Sigma_{JK} \Sigma_{KK}^{-1}(\mathbf{x}_K - \bm{\mu}_K + \Sigma_{KI} G^T \bm{1}_{R_0}), \Sigma_{J|K} \bigr).
\end{align}

\subsection{Derivation of the proposed factor analysis}
The proposed binary factor analysis can be realized as a special case of the proposed distribution.
In factor analysis, observed variables correlate through a continuous latent variable.
We denote the observed binary variables and continuous variables and latent variables as $\mathbf{y}$, $\mathbf{x}$, and $\mathbf{z}$, respectively, and their respective dimensions are $q$, $p_x$, and $p_z$.
In factor analysis, we assume that each binary variable $y_s$ is conditionally uncorrelated.
That is, $\Lambda$ is a diagonal matrix, $\Lambda - I =  \mathrm{diag}(e^{-b_s}), \hspace{0.2cm} (s=1,2,\dots, q)$.

We partition the set of whole indices for the continuous variables in the previous subsection as $I = ((J, L), K) = ((O, L), Z) = (X, Z)$ and redefine the continuous variable itself as $\mathbf{x} = ((\mathbf{x}_J, \mathbf{x}_L), \mathbf{x}_K) \rightarrow ((\mathbf{x}_O, \mathbf{x}_L), \mathbf{z}) = (\mathbf{x}, \mathbf{z})$.
We also partition the set of whole indices for binary variables as $R=((S,U),T) = ((V, U), \emptyset)$, where $\emptyset$ is the empty set.
The index labels of the set of indices $O$ and $V$ may be understood as the initial letters of ``Observed'' and ``Visible'', respectively.
We further put the partitioned matrix and vector as $G_{RI} = (G_{RX}, G_{RZ}) = (0, G_{RZ})$ and $\mathbf{g}_{R I}^T = (\mathbf{g}_{RX}^T, \mathbf{g}_{RZ}^T) = (\bm{0}, \mathbf{g}_{RZ}^T)$ and redefining them as $G_{RZ} \rightarrow G$ and $\mathbf{g}_{RZ}^T \rightarrow \mathbf{g}_R^T$, respectively.
Then, we obtain
\begin{align}
    p(\mathbf{x}, \mathbf{z}, \mathbf{y}=\bm{1}_{R_1}) 
    =&
    \pi_{R_0}(\Sigma_{ZZ}) \, \mathcal{N}( \mathbf{x} \mid \bm{\mu}_X + \Sigma_{XZ} \Sigma_{ZZ}^{-1} (\mathbf{z} - \bm{\mu}_Z), \Sigma_{X|Z} ) 
    \mathcal{N}( \mathbf{z} \mid \bm{\mu}_Z - \Sigma_{ZZ} \mathbf{g}_{R_0}, \Sigma_{ZZ} ), \notag \\
    =&
    \tilde{\pi}_{R_1}(\Sigma_{ZZ}) \, \mathcal{N}( \mathbf{x} \mid \tilde{\bm{\mu}}_X + \Sigma_{XZ} \Sigma_{ZZ}^{-1} (\mathbf{z} - \tilde{\bm{\mu}}_Z), \Sigma_{X|Z} )
    \mathcal{N}( \mathbf{z} \mid \tilde{\bm{\mu}}_Z + \Sigma_{ZZ} G^T \mathbf{y}, \Sigma_{ZZ} ) , \\
    \pi_{R_0}(\Sigma_{ZZ})= &
    \frac{e^{ - \mathbf{b}_{R_0} + \frac{1}{2} \mathbf{g}_{R_0}^T \Sigma_{ZZ} \mathbf{g}_{R_0}} }{\sum_{R_0' \subseteq R} e^{ - \mathbf{b}_{R_0'} + \frac{1}{2} \mathbf{g}_{R_0'}^T \Sigma_{ZZ} \mathbf{g}_{R_0'}}}, \notag \\
    = \tilde{\pi}_{R_1}(\Sigma_{ZZ}) \equiv & 
    \frac{e^{ \bm{1}_{R_1}^T \tilde{\mathbf{b}} + \frac{1}{2} \bm{1}_{R_1}^T G \Sigma_{ZZ} G^T \bm{1}_{R_1}} }{\sum_{R_1' \subseteq R} e^{ \bm{1}_{R_1'}^T \tilde{\mathbf{b}} + \frac{1}{2} \bm{1}_{R_1'}^T G \Sigma_{ZZ} G^T \bm{1}_{R_1'}}},
\end{align}
\begin{align}
    p(\mathbf{z}) =& 
    \sum_{R_0  \subseteq R} \pi_{R_0}(\Sigma_{ZZ}) \, \mathcal{N}( \mathbf{z} \mid \bm{\mu}_Z - \Sigma_{ZZ} \mathbf{g}_{R_0}, \Sigma_{ZZ} ), \notag \\
    =& 
    \sum_{R_1  \subseteq R} \tilde{\pi}_{R_1}(\Sigma_{ZZ}) \, \mathcal{N}( \mathbf{z} \mid \tilde{\bm{\mu}}_Z + \Sigma_{ZZ} G^T \bm{1}_{R_1}, \Sigma_{ZZ} ), \\
    p(\mathbf{x}_O, \mathbf{y}_V|\mathbf{z}) =&
    \frac{\sum_{U_0 \subseteq R \setminus V} e^{ - \mathbf{b}_{R_0} -\mathbf{g}_{R_0}^T (\mathbf{z} -\bm{\mu}_Z)}}{\sum_{R_0' \subseteq R}  e^{ - \mathbf{b}_{R_0'} -\mathbf{g}_{R_0'}^T (\mathbf{z} - \bm{\mu}_Z)}}
     \,
    \mathcal{N}( \mathbf{x}_O \mid \bm{\mu}_O + \Sigma_{OZ} \Sigma_{ZZ}^{-1} (\mathbf{z}-\bm{\mu}_Z), \Sigma_{O|Z} ), \notag \\
    =& 
    \frac{\sum_{U_1 \subseteq R \setminus V} e^{ \bm{1}_{R_1}^T (\tilde{\mathbf{b}} + G (\mathbf{z} - \tilde{\bm{\mu}}_Z))}}{\sum_{R_1' \subseteq R} e^{\bm{1}_{R_1'}^T (\tilde{\mathbf{b}} + G (\mathbf{z} - \tilde{\bm{\mu}}_Z)) }}
     \,
    \mathcal{N}( \mathbf{x}_O \mid \tilde{\bm{\mu}}_O + \Sigma_{OZ} \Sigma_{ZZ}^{-1} (\mathbf{z}-\tilde{\bm{\mu}}_Z), \Sigma_{O|Z} ),
\end{align}
where we have defined the parameters as $\tilde{\bm{\mu}} \equiv \bm{\mu} - \Sigma G^T \bm{1}$ and $\tilde{b} = b - G \Sigma_{ZZ} G^T \bm{1}$.

In a similar way, we partition the set of whole indices for continuous variables in the previous subsection as $I=((J, L), K) = ((Z, L), O)$ and redefine the continuous variable itself as $\mathbf{x} = ((\mathbf{x}_J, \mathbf{x}_L), \mathbf{x}_K) \rightarrow ((\mathbf{z},\mathbf{x}_L),\mathbf{x}_O)$.
We also partition the set of whole indices for binary variables as $R = ((S,U),T) = ((\emptyset, U), V)$.
Putting the partitioned matrix and vector as $G_{RI} = (G_{RZ}, G_{R(L+O)}) = (G_{RZ}, 0)$ and $\mathbf{g}_{R I}^T = (\mathbf{g}_{RZ}^T, \mathbf{g}_{R (L+O)}^T) = (\mathbf{g}_{RZ}^T, \bm{0})$ and redefining them as $G_{RZ} \rightarrow G$ and $\mathbf{g}_{RZ}^T \rightarrow \mathbf{g}_R^T$, we obtain
\begin{align}
    p(\mathbf{x}_O, \mathbf{y}_V) =& 
    \sum_{U_0 \subseteq R \setminus V} \pi_{R_0}(\Sigma_{ZZ}) \, 
    \mathcal{N}( \mathbf{x}_O \mid \bm{\mu}_O - \Sigma_{OZ} \mathbf{g}_{R_0}, \Sigma_{OO} ), \notag \\
    =& 
    \sum_{U_1 \subseteq R \setminus V} \tilde{\pi}_{R_1}(\Sigma_{ZZ}) \,
    \mathcal{N}( \mathbf{x}_O \mid \tilde{\bm{\mu}}_O + \Sigma_{OZ} G^T \bm{1}_{R_1}, \Sigma_{OO} ),
\end{align}
\begin{align}
    p(\mathbf{z}|\mathbf{x}_O, \mathbf{y}_V) =& \frac{\sum_{U_0 \subseteq R \setminus V}  e^{- \mathbf{b}_{R_0} + \frac{1}{2}\mathbf{g}_{R_0}^T \Sigma_{Z|O} \mathbf{g}_{R_0}} e^{-\mathbf{g}_{U_0}^T \Sigma_{ZO} \Sigma_{OO}^{-1} (\mathbf{x}_O - \bm{\mu}_O)}}{\sum_{U_0' \subseteq R \setminus V} e^{ - \mathbf{b}_{R_0'} + \frac{1}{2} \mathbf{g}_{R_0'}^T \Sigma_{Z|O} \mathbf{g}_{R_0'}} e^{-\mathbf{g}_{U_0'}^T \Sigma_{ZO} \Sigma_{OO}^{-1} (\mathbf{x}_O-\bm{\mu}_O)}} \notag \\
    & \mathcal{N}\bigl( \mathbf{z} \mid \bm{\mu}_Z + \Sigma_{ZO}\Sigma_{OO}^{-1}(\mathbf{x}_O - \bm{\mu}_O) - \Sigma_{Z|O} \mathbf{g}_{R_0}, \Sigma_{Z|O} \bigr) , \notag \\
    =& 
    \frac{\sum_{U_1 \subseteq R \setminus V} e^{\bm{1}_{U_1}^T \tilde{\mathbf{b}} + \frac{1}{2} (\bm{1}_{U_1} + \bm{1}_{V_1})^T G \Sigma_{Z|O} G^T (\bm{1}_{U_1} + \bm{1}_{V_1}) + \bm{1}_{U_1}^T G \Sigma_{ZO} \Sigma_{OO}^{-1} (\mathbf{x}_O - \tilde{\bm{\mu}}_O)}}{\sum_{U_1' \subseteq R \setminus V} e^{\bm{1}_{U_1'}^T \tilde{\mathbf{b}} + \frac{1}{2} (\bm{1}_{U_1'} + \bm{1}_{V_1})^T G \Sigma_{Z|O} G^T (\bm{1}_{U_1'} + \bm{1}_{V_1}) + \bm{1}_{U_1'}^T G \Sigma_{ZO} \Sigma_{OO}^{-1} (\mathbf{x}_O - \tilde{\bm{\mu}}_O)}} \notag \\
    &\mathcal{N}\bigl( \mathbf{z} \mid \tilde{\bm{\mu}}_Z + \Sigma_{ZO}\Sigma_{OO}^{-1}(\mathbf{x}_O - \tilde{\bm{\mu}}_O) + \Sigma_{Z|O} G^T (\bm{1}_{U_1} + \bm{1}_{V_1}), \Sigma_{Z|O} \bigr).
\end{align}

If the observed variables have no missing values, the above expressions can be expressed more concisely:
\begin{align}
    p(\mathbf{x}, \mathbf{z}, \mathbf{y}) =& \tilde{\pi}_{R_1}(\Sigma_{ZZ}) \, \mathcal{N}\bigl( (\mathbf{x}^T, \mathbf{z}^T)^T \mid \tilde{\bm{\mu}}_{(X+Z)} + \Sigma_{(X+Z) Z} G \bm{1}_{R_1} \bigr), \notag \\
    =& 
    \tilde{\pi}_{R_1}(\Sigma_{ZZ}) \, 
    \mathcal{N}( \mathbf{x} \mid \tilde{\bm{\mu}}_X + \Sigma_{XZ} \Sigma_{ZZ}^{-1}(\mathbf{z}- \tilde{\bm{\mu}}_Z), \Sigma_{X|Z} ) \, 
    \mathcal{N}( \mathbf{z} \mid \tilde{\bm{\mu}}_Z + \Sigma_{ZZ} G^T \mathbf{y}, \Sigma_{ZZ} ), \\
    p(\mathbf{z}) =& 
    \sum_{R_1  \subseteq R} \tilde{\pi}_{R_1}(\Sigma_{ZZ}) \,
    \mathcal{N}( \mathbf{z} \mid \tilde{\bm{\mu}}_Z + \Sigma_{ZZ} G^T \bm{1}_{R_1}, \Sigma_{ZZ} ), \\
    p(\mathbf{x}, \mathbf{y}|\mathbf{z}) =& 
    \frac{e^{ \mathbf{y}^T (\tilde{\mathbf{b}} + G (\mathbf{z}-\tilde{\bm{\mu}}_Z))}}{\prod_{j=1}^q \bigl( 1 + e^{\tilde{b}_j + \mathbf{g}_j^T (\mathbf{z} - \tilde{\bm{\mu}}_Z)} \bigr) } \, \mathcal{N}\bigl(\mathbf{x} \mid \tilde{\bm{\mu}}_X + \Sigma_{XZ} \Sigma_{ZZ}^{-1} (\mathbf{z}-\tilde{\bm{\mu}}_Z), \Sigma_{X|Z} \bigr), \\
    p(\mathbf{x}, \mathbf{y}) =& \tilde{\pi}_{R_1}(\Sigma_{ZZ}) \, \mathcal{N}(\mathbf{x} \mid \tilde{\bm{\mu}}_X + \Sigma_{XZ} G^T \mathbf{y}, \Sigma_{XX} ), \\
    p(\mathbf{z}|\mathbf{x}, \mathbf{y}) =& \mathcal{N}\bigl(\mathbf{z} \mid \tilde{\bm{\mu}}_Z + \Sigma_{ZX}\Sigma_{XX}^{-1}(\mathbf{x} - \tilde{\bm{\mu}}_X) + \Sigma_{Z|X} G^T \mathbf{y}, \Sigma_{Z|X} \bigr), \\
    \tilde{\pi}_{R_1}(\Sigma_{ZZ}) =&
    \frac{e^{ \bm{1}_{R_1}^T \tilde{\mathbf{b}} + \frac{1}{2} \bm{1}_{R_1}^T G \Sigma_{ZZ} G^T \bm{1}_{R_1}} }{\sum_{R_1' \subseteq R} e^{ \bm{1}_{R_1'}^T \tilde{\mathbf{b}} + \frac{1}{2} \bm{1}_{R_1'}^T G \Sigma_{ZZ} G^T \bm{1}_{R_1'}}}.
\end{align}
To derive the expressions in the main text of the paper, we parameterize the covariance matrix for continuous variables by a block partitioned matrix as follows:
\begin{align}
    \Sigma^{-1} =&
    \begin{bmatrix} \Sigma_{XX} & \Sigma_{XZ} \\ \Sigma_{ZX} & \Sigma_{ZZ} \end{bmatrix}^{-1} 
    \equiv \begin{bmatrix} \Sigma_{x} & \Sigma_{xz} \\ \Sigma_{zx} & \Sigma_{z} \end{bmatrix}^{-1}
    = \begin{bmatrix} \Psi + W \Sigma_z W^T & W \Sigma_z \\ \Sigma_z W^T & \Sigma_z \end{bmatrix}^{-1}, \notag \\
    =& \left(
    \begin{bmatrix}
        \Psi^{1/2} & W \Sigma_z^{1/2} \\ O & \Sigma_z^{1/2} 
    \end{bmatrix}
    \begin{bmatrix}
        \Psi^{1/2} & O \\ \Sigma_z^{1/2} W^T & \Sigma_z^{1/2}
    \end{bmatrix}
    \right)^{-1}
    = \begin{bmatrix} \Psi^{-1} & - \Psi^{-1} W \\ - W^T \Psi^{-1} & \Sigma_z^{-1} + W^T \Psi^{-1} W \end{bmatrix},
\end{align}
where $O$ is a matrix with all elements zero and $\Psi$ is a diagonal matrix with nonnegative diagonal elements.
That is, $\Sigma_{x} = \Psi + W \Sigma_z W^T$, $\Sigma_{xz} = W \Sigma_z$, $\Sigma_{x|z} = \Psi$, $\Sigma_{z|x} = [\Sigma_z^{-1} + W^T \Psi^{-1} W]^{-1}$, where the notation $\Sigma_{x|z} \equiv \Sigma_x - \Sigma_{xz} \Sigma_z^{-1} \Sigma_{zx}$ means the Schur complement.
If we assume that $\Sigma_z$ is a diagonal matrix with nonnegative diagonal elements, then the above expression can be interpreted as a Cholesky decomposition of $\Sigma$, which means that $\Sigma_x$ is by itself a positive semi-definite matrix.
Redefining the parameters and mixing weight as $\tilde{\mathbf{b}} \rightarrow \mathbf{b}$, $\tilde{\bm{\mu}} \rightarrow \bm{\mu}$, $\bm{\mu}=(\bm{\mu}_X, \bm{\mu}_Z) \rightarrow (\bm{\mu}_x, \bm{\mu}_z)$, and $\tilde{\pi}_{R_1} \rightarrow \pi_{R_1}$, we obtain the expressions in the main text of the paper, Eqs.~(\ref{eq:fa_conditional}, \ref{eq:fa_prior}, \ref{eq:fa_induced}, \ref{eq:fa_posterior}, \ref{eq:fa_joint}).

\section{Proof of the identifiability for the proposed factor analysis} \label{sec:appendix_proof}

\subsection{Identifiability for binary variables}
In this appendix, we provide a proposition concerning the identifiability of the proposed binary factor analysis and its proof.

\begin{proposition} \label{proposition:1}
    Let $\mathbf{y}$ be a $p_y$-dimensional column vector of observed binary variables.
    We consider the Ising model where the distribution of $\mathbf{y}$ is parametrized as follows:
    \begin{align}
        p(\mathbf{y}) =& \frac{1}{Z} \exp \biggl\{ \frac{1}{2} \mathbf{y}^T \Sigma_y \, \mathbf{y} \biggr\}, \label{eq:Ising_model} \\
        Z =& \sum_{R_1 \subseteq R} \exp \biggl\{ \frac{1}{2} \bm{1}_{R_1}^T \Sigma_y \, \bm{1}_{R_1} \biggr\}, \\
        \Sigma_y =& 2 B + G G^T, \label{eq:Ising_parameter}
    \end{align}
    where $\Sigma_y$ is an Ising parameter matrix, $B = \mathrm{diag}(\mathbf{b})$ is a $p_y \times p_y$ diagonal matrix, $G$ is a $p_y \times q$ matrix, $q$ is a dimension of the latent space, and $\mathbf{b} + \mathrm{diag}(G G^T) / 2$ and $G G^T$ are bias and interaction terms of the Ising model.
    If the dimension of the latent space satisfies $0 < q < p_y$, the row norms of the matrix $G$ are equal for all features, $\mathrm{diag}(G G^T) = c^2 I$, the symmetric matrix $G^T G$ is diagonalized with its nonzero and nondegenerate eigenvalues sorted in descending order $G^T G = \mathrm{diag}(\omega^2) = \Omega^2$, $\omega_i > \omega_j$, $(i < j)$, and the row sums of the matrix $G$ are nonnegative,
    \begin{align}
        \sum_{i=1}^{p_y} [G]_{ij} \ge 0, \hspace{0.5cm} j \in \{1, 2, \dots, q \}, \label{eq:G_sign}
    \end{align}
    then $\Sigma_{y1} = \Sigma_{y2}$ means $c_1 = c_2$ and $G_1 = G_2$.
\end{proposition}

The minimal and regular exponential family distributions are identifiable with respect to their natural parameters~\cite{Wainwright2008}.
Since the Ising model of the form Eq.~(\ref{eq:Ising_model}) belongs to the minimal and regular exponential family, this Ising distribution itself is identifiable.
Therefore, it suffices to show that the decomposition of the Ising parameter matrix, Eq.~(\ref{eq:Ising_parameter}), is unique.
The parameter of our Ising model can be reexpressed as follows:
\begin{align}
    \Sigma_y \equiv & 2 B + G G^T, \notag \\
    =&
    2 B + \mathrm{diag}(G G^T) + [G G^T - \mathrm{diag}(G G^T)], \notag \\
    =&
    \mathrm{diag}(\Sigma_y) + (G G^T - c^2 I).
\end{align}
The diagonal elements $\mathrm{diag}(\Sigma_y)$ does not depend on the choice of parametrization.
For $G G^T - c^2 I$, the parameter $c$ does not depend on the choice of parametrization by the following lemma.
\begin{lemma} \label{lemma:1}
    Let us consider the matrix $\Delta \Sigma_y \equiv G G^T - \mathrm{diag}(G G^T) = G G^T - c^2 I$.
    Then, for the number of latent dimensions $q < p$, $\Delta \Sigma_{y1} = \Delta \Sigma_{y2}$ means $c_1 = c_2$.
\end{lemma}
\begin{proof}[\textbf{\upshape Proof:}]
    Since the matrix $G G^T$ is symmetric and positive semidefinite, it has a spectral decomposition of the form $G G^T = U \Omega^2 U^T$, where a $p_y \times p_y$ matrix $U$ is the eigenvector matrix, constructed by placing the eigenvectors of $G G^T$ as columns, and satisfies $U^T = U^{-1}$, and $\Omega^2$ is a diagonal matrix whose diagonal elements are eigenvalues of $G G^T$.
    By using the decomposition of the identity matrix $I = U U^T$, the matrix $\Delta \Sigma_y$ also has the following spectral decomposition:
    \begin{align}
        \Delta \Sigma_y = U ( - c^2 I + \Omega^2) U^T.
    \end{align}
    The above equation implies that the eigenvector matrix of $G G^T$ is also the eigenvector matrix of $\Delta \Sigma_y$.
    Since the matrix $G G^T$ is rank-deficient, the eigenvalue diagonal matrix $\Omega^2$ contains zero eigenvalues.
    Then, the smallest eigenvalue of the matrix $\Delta \Sigma_y$ is $-c^2$.
    The eigenvalues of the matrix $\Delta \Sigma_y$ do not depend on the choice of parametrization.
    Thus, the parameter $c$ is uniquely determined regardless of the parametrization.
\end{proof}

From the lemma~\ref{lemma:1} and the relation $\mathrm{diag}(\Sigma_y) = 2 B + c^2 I$, the uniqueness of $B$ holds.
Therefore, if $G_1 G_1^T = G_2 G_2^T$ implies $G_1 = G_2$, the Ising model of the form Eq.~(\ref{eq:Ising_model}) is identifiable.
The uniqueness of $G$ is demonstrated by the following lemma.
\begin{lemma} \label{lemma:2}
    For $p_y \times q$ matrix $G$, if $q \le p$, the symmetric matrix $G^T G$ is diagonalized with its nonzero and nondegenerate eigenvalues sorted in descending order, and the row sums of the matrix $G$ are nonnegative, then $G_1 G_1^T = G_2 G_2^T$ means $G_1 = G_2$.
\end{lemma}
\begin{proof}[\textbf{\upshape Proof:}]
    The diagonalization condition of $G^T G$ means that each column of $G$ is orthogonal.
    We separate the norms of the column vectors of $G$ and express it as $G = (G \Omega^{-1}) \Omega$, where $\Omega$ is a $q \times q$ diagonal matrix with positive diagonal elements, and each column of $G \Omega^{-1}$ is orthonormalized.
    Then, the lemma~\ref{lemma:2} is equivalent to the statement: if
    \begin{align}
        (G_1 \Omega_1^{-1}) \Omega_1^2 (\Omega_1^{-1} G_1^T) = (G_2 \Omega_2^{-1}) \Omega_2^2 (\Omega_2^{-1} G_2^T),
    \end{align}
    then $G_1 = G_2$ and $\Omega_1 = \Omega_2$.
    The above equation is a truncated spectral decomposition of the symmetric positive semidefinite matrix, and the spectral decomposition is unique up to the ordering of eigenvalues and the sign of eigenvectors.
    Therefore, if the nonzero eigenvalues of $G G^T$ are nondegenerate and sorted in descending order, and the sign of the column of $G$ satisfies Eq.~(\ref{eq:G_sign}), then the decomposition is unique.
    Therefore, $G_1 = G_2$.
\end{proof}

This completes the proof of the proposition~\ref{proposition:1}.

\subsection{Identifiability for mixed continuous and binary variables} \label{sec:appendix_proof_2}
In this appendix, we provide a proof of the theorem concerning the identifiability of the proposed factor analysis with mixed-type variables.

The observed distribution in the proposed factor analysis corresponds to a parametrized location model.
In the location model, it is assumed that the continuous variables follow a normal distribution for each composite category $\mathbf{y}$, where the composite category is constructed from multiple dummy variables $(y_1, y_2, \dots, y_q)$.
To be more precise, the observed distribution is expressed as follows~\cite{Olkin1961}:
\begin{align}
    p(\mathbf{x}, \mathbf{y} = \bm{1}_{R_1}) =& p(\mathbf{x} \mid \mathbf{y} = \bm{1}_{R_1}) p(\mathbf{y} = \bm{1}_{R_1}), \notag \\
    =&
    \pi_{R_1} \mathcal{N}(\mathbf{x} \mid \bm{\mu}_x^{R_1}, \Sigma_x^{R_1}),
\end{align}
where $\pi_{R_1}$ is a class probability corresponding to each dummy-variable pattern (composite category) and satisfies $\sum_{R_1 \subseteq R} \pi_{R_1} = 1$, and $\bm{\mu}_x^{R_1}$ and $\Sigma_x^{R_1}$ are location and covariance parameters corresponding to the composite category, respectively.
The proposed factor analysis corresponds to a location model with the following parameters:
\begin{align}
    \bm{\mu}_x^{R_1} =& \bm{\mu}_x + W G^T \bm{1}_{R_1}, \\
    \Sigma_x^{R_1} =& \Sigma_x
    = \Psi + W W^T, \\
    \pi_{R_1} =& 
    \frac{1}{Z} \exp \biggl\{ \frac{1}{2} \bm{1}_{R_1}^T \Sigma_y \bm{1}_{R_1} \biggr\}, \\
    Z =& \sum_{R_1' \subseteq R} \exp \biggl\{ \frac{1}{2} \bm{1}_{R_1'}^T \Sigma_y \bm{1}_{R_1'} \biggr\}, \\
    \Sigma_y =& 2 B + G G^T,
\end{align}
where we have set $\Sigma_z = I$.
In other words, the proposed factor analysis corresponds to a location model in which class probabilities are expressed as an Ising model, location parameters are shifted by dummy variables, and covariance parameters are expressed as a low-rank perturbation of a diagonal matrix $\Psi$ as in the case of factor analysis with continuous variables.

The location model is trivially identifiable, provided that all class probabilities are strictly positive and that the location parameters are distinct across classes.
When the class probabilities are given by the Ising distribution (location model with Ising-distributed class probabilities) and the number of latent dimensions satisfies $q < p_y$, the discussion up to the lemma~\ref{lemma:1} implies that the row norm of the binary factor loading matrix $G$ does not depend on the parametrization.
In this case, from the relation $\mathrm{diag}(\Sigma_y) = 2 B + c^2 I$, the parameter $B$ is also uniquely determined.
Here, we consider the dimensionless parameters by dividing the factor loading matrices by $\Psi^{-1/2}$ as follows:
\begin{align}
    \tilde{\bm{\mu}} \equiv & \Psi^{-1/2} \bm{\mu} = \Psi^{-1/2} \bm{\mu}_x + c^2 \hat{W} \hat{G}^T \mathbf{y}, \\
    \tilde{\Sigma}_x \equiv & \Psi^{-1/2} \Sigma_x \Psi^{-1/2} =  I + c^2 \hat{W} \hat{W}^T, \\
    \Sigma_y =& 2 B + c^2 \hat{G} \hat{G}^T.
\end{align}
These parameters can be collectively expressed using a block matrix as follows:
\begin{align}
    \begin{bmatrix} \tilde{\Sigma}_x & c^2 \hat{W} \hat{G}^T \\ c^2 \hat{G} \hat{W}^T & \Sigma_y \end{bmatrix}
    =&
    \begin{bmatrix} I + c^2 \hat{W} \hat{W}^T & c^2 \hat{W} \hat{G}^T \\ c^2 \hat{G} \hat{W}^T & 2 B + c^2 \hat{G} \hat{G}^T \end{bmatrix}, \notag \\
    =&
    \begin{bmatrix} I & O \\ O & 2 B \end{bmatrix} + 
    c^2 \begin{bmatrix} \hat{W} \hat{W}^T  & \hat{W} \hat{G}^T \\
    \hat{G} \hat{W}^T & \hat{G} \hat{G}^T \end{bmatrix}, \notag \\
    =&
    \begin{bmatrix} I & O \\ O & 2 B \end{bmatrix} + 
    c^2 \begin{bmatrix} \hat{W} \\ \hat{G} \end{bmatrix}
    \begin{bmatrix} \hat{W}^T & \hat{G}^T \end{bmatrix}, \notag \\
    \equiv&
    \begin{bmatrix} I & O \\ O & 2 B \end{bmatrix} + c^2 \hat{M} \hat{M}^T.
\end{align}
Therefore, the identifiability of the model can be reexpressed in terms of the matrix $\hat{M}$ of the form $\hat{M} \hat{M}^T$. 
From the lemma~\ref{lemma:2}, if the dimension of the latent space satisfies $q \le p_x + p_y$, the row norms of the normalized factor loading matrix $\hat{M}$ are equal for all features, $\mathrm{diag}(\hat{M} \hat{M}^T) = 1$, the symmetric matrix $\hat{M}^T \hat{M}$ is diagonalized with its nonzero and nondegenerate eigenvalues sorted in descending order, and the row sums of the matrix $\hat{M}$ are nonnegative, then $\hat{M}_1 \hat{M}_1^T = \hat{M}_2 \hat{M}_2^T$ means $\hat{M}_1 = \hat{M}_2$.

The uniqueness of the parameter $\bm{\mu}_x$ is immediately apparent.
This completes the proof of the theorem~\ref{theorem:1}.

\bibliographystyle{unsrt}
\bibliography{grassmann_v5}

\end{document}